\pdfoutput=1 

\documentclass[a4paper,USenglish,cleveref,autoref,thm-restate,authorcolumns,runningheads,envcountsame]{lipics-v2019}
\nolinenumbers
\hideLIPIcs
\usepackage{amsmath,amssymb,mathtools,bbm}
\usepackage{xcolor}
\usepackage{caption}
\usepackage{subcaption}
\usepackage{booktabs}
\usepackage[version=0.96]{pgf}
\usepackage{tikz}
\usetikzlibrary{arrows,calc,fit,shapes,automata,backgrounds,positioning}
\usepackage{pifont}
\newcommand{\cmark}{\ding{51}}
\newcommand{\xmark}{{\color{gray}\ding{55}}}

\usepackage{sty/colors}      
\usepackage{sty/utilities}   
\usepackage{sty/notation}    
\usepackage{thm-restate}     

\hypersetup{hidelinks}

\title{A Classification of Weak Asynchronous Models of Distributed Computing}

\author{Javier Esparza}{Technische Universit\"{a}t M\"{u}nchen, Germany}{esparza@in.tum.de}{https://orcid.org/0000-0001-9862-4919}{}

\author{Fabian Reiter}{LIGM, Universit\'e Gustave Eiffel, France}{fabian.reiter@gmail.com}{https://orcid.org/0000-0003-1268-4107}{}

\authorrunning{J.\,Esparza and F.\,Reiter}

\Copyright{Javier Esparza and Fabian Reiter}



\acknowledgements{The authors thank Ahmed Bouajjani for many interesting discussions,
  and several anonymous reviewers for their helpful feedback.}

\funding{This work was supported by the
  \href{https://paves.model.in.tum.de}
  {ERC project PaVeS (Advanced Grant 787367)}.}

\relatedversion{To appear in the proceedings of CONCUR 2020 (published by LIPIcs).}

\begin{document}

\maketitle

\begin{abstract}
We conduct a systematic study of asynchronous models of distributed computing
consisting of identical finite-state devices that cooperate in a network
to decide if the network satisfies a given graph-theoretical property.
Models discussed in the literature differ in
the detection capabilities of the agents residing at the nodes of the network
(detecting the set of states of their neighbors,
or counting the number of neighbors in each state),
the notion of acceptance
(acceptance by halting in a particular configuration, or by stable consensus),
the notion of step
(synchronous move, interleaving, or arbitrary timing),
and the fairness assumptions
(non-starving, or stochastic-like).
We study the expressive power of the combinations of these features,
and show that the initially twenty possible combinations
fit into seven equivalence classes.
The classification is the consequence of several equi-expressivity results
with a clear interpretation. In particular, we show that
acceptance by halting configuration only has non-trivial expressive power
if it is combined with counting, and that synchronous and interleaving models have the same power
as those in which an arbitrary set of nodes can move at the same time.
We also identify simple graph properties that distinguish
the expressive power of the seven classes.
\end{abstract}

\section{Introduction}
\label{sec:intro}

Distributed computing is increasingly interested in the study of networks 
of natural or artificial devices, like molecules, cells, microorganisms, or nano-robots.
These devices have very limited computational and communication capabilities, and are  indistinguishable. In particular, a device cannot recognize whether its current communication 
partner is the same as a past one. This stands in stark contrast to the devices of standard computer 
networks, which has motivated researchers to question the suitability of traditional distributed computing models for the study of these networks, and to propose new ones. Examples include
population protocols \cite{AADFP06,angluin2005stably}, chemical reaction networks \cite{SoloveichikCWB08}, networked finite state machines \cite{EmekW13}, the weak models of distributed computing of \cite{HJKLLLSV15}, and the beeping model \cite{CK10}. A survey discussing many of them, and more, can be found in \cite{NB15}. 

All these models share several common features, introduced to capture the limitations of the devices \cite{EmekW13}:  the network can have an arbitrary topology; all nodes of the network have a finite number of states, independent of the size of the network or its topology; all nodes run the same protocol; and state changes only depend on the states of a bounded number of neighbors, again independent of the size of the network.

Unfortunately, despite this very substantial common ground, the models still differ in many aspects, which makes it hard to compare results across them, or decide which features are essential for a particular result. A study of the models allows one to identify four specific junctions at which they choose different paths:
\begin{itemize}
\item \textit{Detection}. In some models, agents can only detect the \emph{existence} of neighbors in  a certain state \cite{HJKLLLSV15}. In others, they can \emph{count} their number, up to a fixed threshold \cite{EmekW13,HJKLLLSV15}.  For example, in biological models, cells communicate by emitting special kinds of proteins, and detecting them; in some models the cells may detect the presence of the protein when its concentration exceeds a given threshold, while in others they are able to detect different concentration levels.
\item \textit{Acceptance}. Some models compute by \emph{stable consensus}, which requires all nodes to eventually agree on the outcome of the computation (but the nodes do not need to \emph{know} that consensus has been reached) \cite{AADFP06,angluin2005stably,SoloveichikCWB08}, while others require the nodes to reach a consensus in a halting configuration \cite{HJKLLLSV15}. Acceptance by stable consensus is computationally powerful, since it permits the algorithm designer to concentrate on ensuring that every bad input is eventually rejected; declaring all non-rejecting states accepting ensures that every good input is eventually accepted. 
\item \textit{Selection}. In some models, at each step a scheduler chooses an
  arbitrary set of nodes to make a step \cite{EmekW13,Reiter17}, while in
  others it is exactly one, or exactly one pair of neighboring nodes
  \cite{AADFP06,angluin2005stably, SoloveichikCWB08}. We call the latter \emph{exclusive} or \emph{interleaving} models.
Intuitively, interleaving models are useful when it can be assumed that process steps are much faster than the time interval between them, while the former policy does not need this assumption. In addition, they help the algorithm designer, who can assume that agents act in mutual exclusion.
(Examples where this is useful can be found in the proofs of Propositions~\ref{prop:CasF-recog-stars} and~\ref{prop:population-protocols}.)
Another common option for selection is the \emph{synchronous} execution model  \cite{HJKLLLSV15},
where all nodes are selected in each step.
Again this can be helpful for designing algorithms,
but it is incompatible with exclusive selection.
\item \textit{Fairness}. Some models use fairness assumptions designed to model or approximate stochastic behavior \cite{AADFP06,angluin2005stably,SoloveichikCWB08}, while others choose minimal notions, like ``all nodes make a step infinitely often'', which only assume the absence of crash faults
(see, e.g., \cite{Francez,LehmannPS81}).
Stochastic-like assumptions are reasonable for biological or chemical models, but can be too strong for networks of artificial nodes, which may follow non-random execution policies. Stochastic models may be able to solve problems that cannot be solved with weaker fairness assumptions.
\end{itemize}
The goal of this paper is to explore the space of models spanned by the above parameters, and compare their computational power within a specific framework. For this  we use {\em distributed automata}, a generic formalism for the description of finite-state distributed algorithms. Such an automaton consists of a set of rules that tell the nodes of a graph how to change their state depending on the states of their neighbors. Intuitively, the automaton describes an algorithm that allows the nodes of an input graph to decide, in a distributed way, whether the graph satisfies a given property. The computational power of a class of distributed automata is then given by the class of graph languages recognized by the automata in the class, or, in other words, by the graph properties that the class of automata can decide.

We start with twenty classes of distributed automata, and show that with respect to their computational power, they fall into seven different classes. This reduction is a consequence of two results presented in this paper: (1) acceptance by halting configuration only has non-trivial expressive power if it is  combined with counting; (2) both interleaving and synchronous selection have the same power as liberal selection where arbitrarily many nodes can move at the same time (and therefore, one can design an automaton in an interleaving or synchronous model, which is less error prone, and then translate it to a liberal model). Some of the simulations we design to prove the results are of independent interest. In particular, we give explicit constructions showing how to simulate interleaving models by non-interleaving~ones.  

The paper is organized as follows. Section \ref{sec:prelim} introduces distributed automata and their variants. Sections \ref{sec:trivial} to \ref{sec:exclusive} show that the variants collapse to at most the seven equivalence classes mentioned above. Section \ref{sec:separations} contains separation results showing that the seven classes are different. Finally, Section~\ref{sec:pprotocols} presents further results on their expressive power. Proofs missing or only sketched in the main text can be found in the Appendix.

%


\section{A taxonomy of distributed automata}
\label{sec:prelim}

Given sets $X, Y$,
we denote by $\PowerSet{X}$ the power set of $X$,
and by $X^Y$ the set of functions~${Y \to X}$.
We define
$\Range[m]{n} \DefEq \SetBuilder{i \in \Integers}{m \leq i \leq n}$
and
$\Range{n} \DefEq \Range[0]{n}$,
for any $m, n \in \Integers$ such that $m \leq n$.
Angle brackets indicate excluded endpoints, e.g.,
$\Range*[m]{n} \DefEq \Range[m - 1]{n}$
and
$\Range{n}* \DefEq \Range[0]{n-1}$.

Let~$\Alphabet$ be a finite set.
A \intro{($\Alphabet$-labeled, undirected) graph} is a triple
$G = \Tuple{V, E, \lambda}$,
where
$V$~is a finite nonempty set of \intro{nodes},
$E$~is a set of undirected \intro{edges} of the form
$e = \Set{u, v} \subseteq V$
such that $u \neq v$,
and
$\lambda \colon V \to \Alphabet$
is a \intro{labeling}.
Isomorphic graphs are considered to be equal.
\textbf{Convention}: Throughout the paper,
all graphs have at least two nodes and are connected.

\subsection{Distributed automata}

Distributed automata take a graph as input,
and either accept or reject it.
To define them we first introduce distributed machines.

\smallskip\noindent\textbf{Distributed machines.}
Let~$\Alphabet$ be a finite set of symbols
and let $\beta \in \Positives$.
A \intro{(distributed) machine}
with \emph{input alphabet}~$\Alphabet$ and \emph{counting bound}~$\beta$
is a tuple
$M~=~\Tuple{Q, \delta_0, \delta, Y, N}$,
where
$Q$~is a finite set of \intro{states},
$\delta_0 \colon \Alphabet \to Q$
is an \intro{initialization function},
$\delta \colon Q \times \Range{\beta}^{Q} \to Q$
is a \intro{transition function},
and $Y, N \subseteq Q$ are two sets
of \emph{accepting} and \emph{rejecting} states,
respectively.
The function~$\delta$ updates the state of a node~$v$
based on the number of neighbors $v$ has in each state,
but it can only detect
if $v$ has $0, 1, \ldots, (\beta-1)$,
or at least $\beta$ neighbors in a given state.


\smallskip\noindent\textbf{Selections, schedules, configurations, runs, and acceptance.}
A \emph{selection} of a $\Alphabet$-labeled graph
$G = \Tuple{V, E, \lambda}$ is a set~$\Selection \subseteq V$,
and a \emph{schedule} of~$G$
is an infinite  sequence of selections
$\Schedule = (\Selection_0, \Selection_1, \Selection_2, \ldots) \in (2^V)^\omega$.
Intuitively,
the selection~$\Selection_t$ is the set of nodes activated by the scheduler at time $t$.

Let
$M = \Tuple{Q, \delta_0, \delta, Y, N}$
be a distributed machine with input alphabet $\Alphabet$.
A \intro{configuration} of $M$ on~$G$ is a mapping
$\Configuration \colon V \to Q$.
Given a configuration $\Configuration$ and a node $v \in V$,
we let
$N_v^{\Configuration} \colon Q \to \Range{\beta}$
denote the function that assigns to each state $q$
the number of neighbors of~$v$ that are in state~$q$
up to threshold~$\beta$, i.e.,
$\min \bigSet{
  \beta,\,
  \Card{\SetBuilder{u}{\Set{u, v} \in E \land \Configuration(u) = q}}}$.
We call~$N_v^{\Configuration}$
the \intro{$\beta$-bounded multiset} of states of $v$'s neighbors.

For any selection~$\Selection$,
we define the \intro{successor configuration}
of~$\Configuration$ via~$\Selection$
to be the configuration
$\mathit{succ}_{\delta}(\Configuration, \Selection)$
that one obtains from~$\Configuration$
if all nodes in~$\Selection$
evaluate the transition function~$\delta$ simultaneously
while the remaining nodes keep their current state.
Formally,
for all $v \in V$,
\begin{align*}
  \mathit{succ}_{\delta}(\Configuration, \Selection)(v) =
  \begin{cases*}
    \Configuration(v)
    & if $v \notin \Selection$ \\
    \delta
    \bigl(
    \Configuration(v), N_v^{\Configuration}
    \bigr)
    & if $v \in \Selection$.
  \end{cases*}
\end{align*}

This brings us directly to the notion of a run.
Given a schedule
$\Schedule = (\Selection_0, \Selection_1, \Selection_2, \ldots)$,
the \intro{run} of~$M$ on~$G$ scheduled by~$\Schedule$
is the infinite sequence
$\Run = \Tuple{\Configuration_0, \Configuration_1, \Configuration_2, \dots}$
of configurations that are defined inductively as follows,
where~$\circ$ denotes function composition,
and $t \in \Naturals$:
\begin{align*}
  \Configuration_0 = \delta_0 \circ \lambda
  \qquad \text{and} \qquad
  \Configuration_{t+1} = \mathit{succ}_{\delta}(\Configuration_t, \Selection_t).
\end{align*}

A configuration $\Configuration$ is \intro{accepting} if $\Configuration(v) \in Y$ for every $v \in V$, and \intro{rejecting} if $\Configuration(v) \in N$ for every $v \in V$.  
A run~$\Run = \Tuple{\Configuration_0, \Configuration_1, \Configuration_2, \dots}$ of~$M$ on~$G$ is \intro{accepting} if there is a  time $t \in \Naturals$ such that $\Configuration_{t'}$ is accepting for every $t' \geq t$. In other words, a run is accepting if from some time on it only visits accepting configurations. Similarly, $\Run$ is \intro{rejecting} if eventually all visited configurations are rejecting. Following \cite{AADFP06}, we call this \emph{acceptance by stable consensus}.

\smallskip\noindent\textbf{Distributed automata.} Not every schedule of a distributed machine models an execution; for example, schedules in which a node is never activated are usually considered illegal.  We assume that distributed machines are controlled by a scheduler that ensures that the machine executes a legal run. Formally, a \emph{scheduler} is a pair $\Scheduler = (\SelectionConstr,\FairnessConstr)$, where $\SelectionConstr$ is a \intro{selection constraint} that assigns to every graph  $G=\Tuple{V, E, \lambda}$ a set~$\SelectionConstr(G) \subseteq 2^V$ of \emph{permitted selections} such that every node $v \in V$ occurs in at least one selection $\Selection \in \SelectionConstr(G)$, and $\FairnessConstr$ is a \emph{fairness constraint} that assigns to every graph $G$ a set $\FairnessConstr(G) \subseteq \SelectionConstr(G)^\omega$ of \intro{fair schedules} of $G$.
We call the runs with schedules in $\FairnessConstr(G)$ \emph{fair runs}
(with respect to~$\Scheduler$).

A \emph{distributed automaton} is a pair $A=(M,\Scheduler)$, where $M$ is a machine 
and $\Scheduler$ is a scheduler satisfying the \emph{consistency condition}: for every graph $G$, either all fair runs of $M$ on $G$ are accepting, or all fair runs of $M$ on $G$ are rejecting.
Intuitively,
the machine is “immune” to the scheduler
because its answer is independent of the scheduler's choices.
This formalizes the standard notion of “asynchronous distributed algorithm”.
Notice that the consistency condition is a very strong \emph{semantic} requirement.
Although we will not do so in this paper,
one can prove that it is undecidable
whether a given pair $\Tuple{M, \Scheduler}$ satisfies it.

$A$ \intro{accepts} $G$ if every fair run of $A$ on $G$ is accepting, and \intro{rejects} $G$ otherwise. The language $L(A)$ recognized by $A$ is the set of graphs it accepts.
Two automata are equivalent if they recognize the same language.

\subsection{Classifying distributed automata.} 

We classify automata according to four criteria: detection capabilities, acceptance condition,  selection constraint, and fairness constraint. The first two criteria concern the distributed machine, and the other two the scheduler.  For each criterion, we investigate some of the major options that have been considered in the literature.

\smallskip\noindent \textbf{Detection.} In some models, agents can only detect the existence of neighbors in a certain state. This corresponds to \emph{non-counting machines}, i.e., machines with counting bound $\beta = 1$.  Other models can detect the number of neighbors up to a higher bound \cite{HJKLLLSV15}.

\smallskip\noindent \textbf{Acceptance.} As mentioned above, distributed machines accept by stable consensus. This is the acceptance condition of population protocols and chemical reaction networks \cite{AADFP06,angluin2005stably, SoloveichikCWB08}. Other models consider a notion of acceptance where each node explicitly decides to accept or reject \cite{HJKLLLSV15}. This notion is captured by halting automata. A machine~$M$ is \emph{halting} if its transition function does not allow the nodes to leave accepting or rejecting states, i.e., if $\delta(q, P) = q$ for every $q \in Y \cup N$ and every $\beta$-bounded multiset $P \in \Range{\beta}^{Q}$. In halting machines, each node knows whether the input graph will be accepted the moment it enters an accepting or rejecting state. Indeed, by the consistency condition, in every fair run, eventually either all nodes occupy accepting states, or all nodes occupy rejecting states. Since nodes can never leave an accepting state once they enter it, each node that enters such a state knows that all other nodes will eventually do likewise. The same applies to rejecting states.

\smallskip\noindent \textbf{Selection.} A scheduler  $\Scheduler = (\SelectionConstr,\FairnessConstr)$ is \emph{synchronous} on $G=\Tuple{V, E, \lambda}$ if $\SelectionConstr(G) = \Set{V}$. Intuitively, at every step all nodes make a move. $\Scheduler$ is \emph{exclusive} or \emph{interleaving-based} on $G$ if $\SelectionConstr(G) = \{ \{v\} \mid v \in V \}$.
Intuitively, at every step exactly one node makes a move, i.e., nodes execute steps in mutual exclusion.
Finally, $\Scheduler$ is \emph{liberal} on~$G$ if $\SelectionConstr(G) = 2^V$. Intuitively, at every step an arbitrary subset of nodes makes a move.
A scheduler is called \emph{synchronous} if it is synchronous on every graph. Exclusive and liberal schedulers are defined analogously.


\smallskip\noindent \textbf{Fairness.} A schedule $\Schedule = (\Selection_0, \Selection_1, \ldots)$ of a graph $G$ is \emph{weakly fair} if for every node~$v$ of~$G$, there exist infinitely many indices $t$ such that $v \in \Selection_t$. 
In other words, a schedule is weakly fair if every node is active infinitely often. 
A scheduler $\Scheduler =(\SelectionConstr,\FairnessConstr)$ is  \emph{weakly fair} if $\FairnessConstr(G)$ contains precisely the weakly-fair schedules of $\SelectionConstr(G)^{\omega}$ for every graph $G$.  This is the weakest fairness constraint one can impose on distributed automata; it only excludes runs in 
which a node crashes, and does not participate in the computation anymore. 

With respect to a given selection constraint~$\SelectionConstr$, a schedule $\Schedule = (\Selection_0, \Selection_1, \ldots) \in \SelectionConstr(G)^\omega$ of a graph $G$ is \emph{strongly fair} if for every finite sequence $(\Selection[2]_0, \ldots, \Selection[2]_n) \in \SelectionConstr(G)^*$ there exist infinitely many indices $t$ such that $(\Selection[1]_t, \Selection[1]_{t+1}, \ldots, \Selection[1]_{t+n}) = (\Selection[2]_0, \Selection[2]_1, \ldots, \Selection[2]_n)$.  Intuitively, strong fairness requires that every possible finite sequence of selections is scheduled infinitely often. 
If every node is selected independently with positive probability, stochastic schedules are almost surely strongly fair. A scheduler $\Scheduler =(\SelectionConstr,\FairnessConstr)$ is \emph{strongly fair} if for every graph $G$, the set $\FairnessConstr(G)$ contains precisely the strongly-fair schedules of $\SelectionConstr(G)^\omega$. 

\begin{remark}
\label{rem:fairness}
Whether a schedule $\sigma$ of a graph $G=\Tuple{V, E, \lambda}$ is strongly fair or not depends on $\SelectionConstr(G)$. For example,  if $\SelectionConstr(G) = \{ V \}$, then the synchronous schedule $V^\omega$ is strongly fair, but if $\SelectionConstr(G) = 2^V$, then it is not. 
\end{remark}

Our notion of strong fairness implies an apparently stronger one, used frequently in the literature, stating that
in a strongly fair run, a sequence of configurations that is enabled infinitely often must occur infinitely often:

\begin{restatable}{lemma}{LemStrongFairness}
\label{lem:strongfairness}
  Let~$A$ be a strongly fair automaton
  and
  $\Tuple{\Configuration[2]_0, \dots, \Configuration[2]_n}$
  be a sequence of configurations of~$A$
  such that
  $\Configuration[2]_{i+1}$ is the successor configuration of~$\Configuration[2]_i$
  via some selection~$\Selection_i$ permitted by~$A$,
  for $i \in \Range[0]{n}*$.
  For any fair run
  $\Run = \Tuple{\Configuration[1]_0, \Configuration[1]_1, \dots}$
  of~$A$,
  if $\Configuration[1]_i = \Configuration[2]_0$
  for infinitely many indices $i \in \Naturals$,
  then
  $\Tuple{\Configuration[1]_j, \dots, \Configuration[1]_{j+n}} =
  \Tuple{\Configuration[2]_0, \dots, \Configuration[2]_n}$
  for infinitely many indices $j \in \Naturals$.
\end{restatable}

The classification above yields 24 classes of automata (four classes of machines and six classes of schedulers). 
To assign mnemonics to them, we use lowercase letters for the most restrictive machine variants (i.e., non-counting and halting), and the same letters in uppercase for the other variants. With schedulers we proceed the other way round, assigning lowercase letters to the most liberal variants (i.e., liberal selection and weak fairness). Intuitively, due to the consistency condition, the more liberal a scheduler, the harder it is for an automaton to recognize a graph language, because more runs have to yield the same result. So, loosely speaking, we expect the expressive power to increase with the number of uppercase letters.

\begin{center}
  \begin{tabular}{l@{\hspace{5ex}}l@{\hspace{5ex}}l@{\hspace{5ex}}l}
    \emph{Detection}                   & \emph{Acceptance}                              & \emph{Selection}                         & \emph{Fairness} \\
    \midrule
    \DetectionType{set}: non-counting  & \AcceptanceType{halting}: halting              & \SelectionType{liberal}: liberal         & \FairnessType{weak}: weak \\
    \DetectionType{multiset}: counting & \AcceptanceType{stabilizing}: stable consensus & \SelectionType{exclusive}: exclusive     & \FairnessType{strong}: strong \\
                                       &                                                & \SelectionType{synchronous}: synchronous &
  \end{tabular}
\end{center}

\newcommand{\res}[1]{\textbf{#1.}}
We denote each class of automata by a string
$wxyz \in
\Set{\DetectionType{set},\DetectionType{multiset}}
\times
\Set{\AcceptanceType{halting},\AcceptanceType{stabilizing}}
\times
\Set{\SelectionType{liberal},\SelectionType{exclusive},\SelectionType{synchronous}}
\times
\Set{\FairnessType{weak},\FairnessType{strong}}$.
The class of languages recognized by $wxyz$-automata is denoted
$\LanguageClass(wxyz)$.
The following lemma states all relations between language classes
that follow directly from the definitions.
Statement~1 abbreviates
``$\LanguageClass(\DetectionType{set}xyz) \subseteq  \LanguageClass(\DetectionType{multiset}xyz)$
for all
$x \in \{\AcceptanceType{halting},\AcceptanceType{stabilizing} \}$,
$y \in
\{\SelectionType{liberal},\SelectionType{exclusive},\SelectionType{synchronous}\}$,
$z \in \{\FairnessType{weak},\FairnessType{strong}\}$''.
We use the same convention in Statements~2 to~5,
and throughout the paper.
That is,
any statement with
four-letter strings containing the wildcard symbol~\DetectionType{*}
must be expanded into the list of all statements
that can be obtained by replacing
identically positioned occurrences of~\DetectionType{*}
with the same letter.

\begin{restatable}{lemma}{LemSimple}\label{lem:simple}
\res{1}~$\LanguageClass(\Type{set}{*}{*}{*}) \subseteq \LanguageClass(\Type{multiset}{*}{*}{*})$,\quad
\res{2}~$\LanguageClass(\Type{*}{halting}{*}{*}) \subseteq \LanguageClass(\Type{*}{stabilizing}{*}{*})$,\quad
\res{3}~$\LanguageClass(\Type{*}{*}{*}{weak}) \subseteq \LanguageClass(\Type{*}{*}{*}{strong})$,\quad
\res{4}~$\LanguageClass(\Type{*}{*}{liberal}{weak}) \subseteq \LanguageClass(\Type{*}{*}{exclusive}{weak})$,\quad
\res{5}~$\LanguageClass(\Type{*}{*}{liberal}{weak}) \subseteq \LanguageClass(\Type{*}{*}{synchronous}{weak})$,\quad
\res{6}~$\LanguageClass(\Type{*}{*}{synchronous}{strong}) \subseteq \LanguageClass(\Type{*}{*}{synchronous}{weak})$.
\end{restatable}

Lemma~\ref{lem:simple} leads to the diagram in Figure~\ref{fig:lattice}, showing 20 automata classes (we have $\LanguageClass(\Type{*}{*}{synchronous}{weak}) = \LanguageClass(\Type{*}{*}{synchronous}{strong})$ by Statements~3 and~6).  An arrow between two classes means that every graph language recognized by the source class is also recognized by the target class.

The reader probably finds Figure~\ref{fig:lattice} very complicated. We also do, and this was the motivation for the present paper. How many of these classes are really different? In the next sections we show that classes with the same color have the same expressivity, and thus that the diagram of Figure~\ref{fig:lattice} collapses to the one of Figure \ref{fig:lattice-CAF}, which contains only seven classes.

\begin{figure}[htb]
\begin{center}
\includegraphics{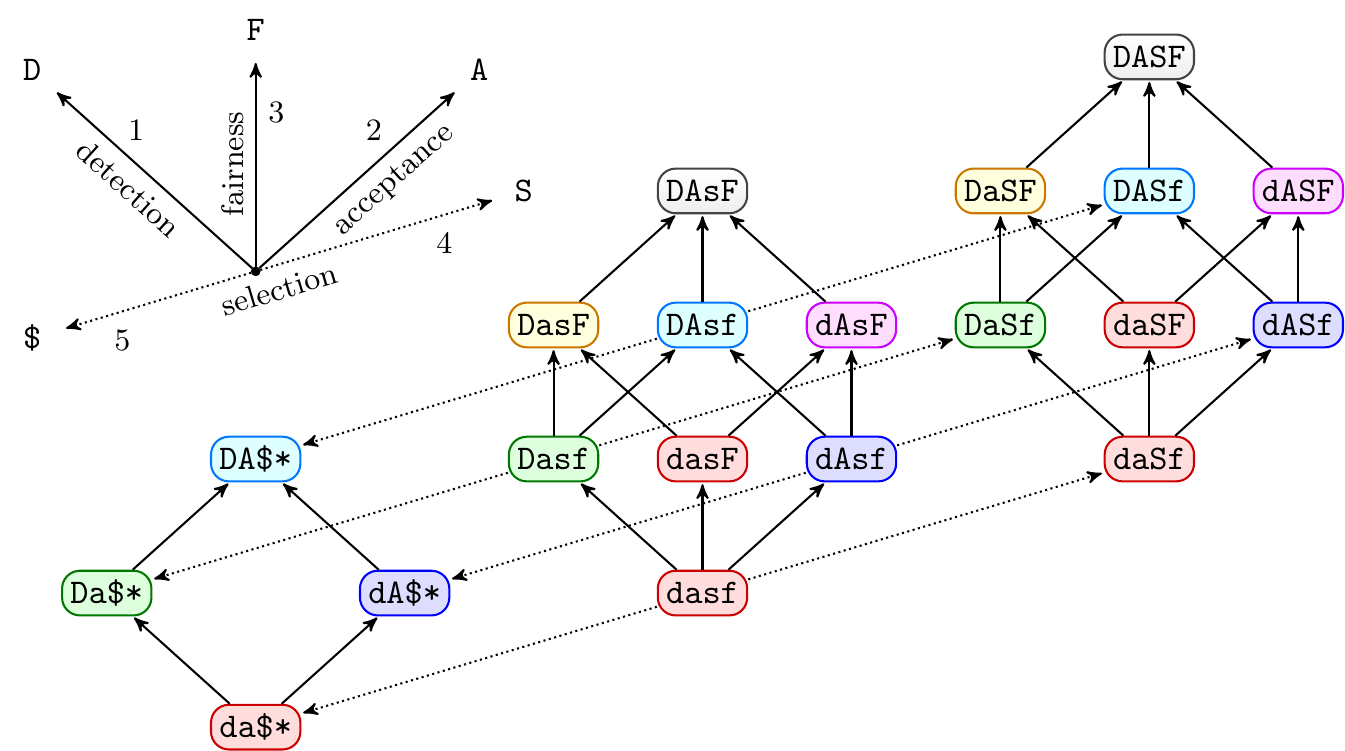}
\end{center}
\vspace{-2ex}
\caption{Initial classification of the models according to the class of graph languages they recognize. Arrows indicate inclusion between classes of languages. The diagram can be thought of as lying in four-dimensional space, where each dimension represents one of our four parameters. The vectors of the “coordinate system” are labeled with the statement number of Lemma \ref{lem:simple} that proves the inclusions in the corresponding direction. In the coming sections, classes are shown to be equal if and only if they have the same color, reducing the 20 classes to 7,  as shown in Figure \ref{fig:lattice-CAF}. This means in particular that we completely eliminate the dimension of selection (shown in dotted lines), leaving us with only three dimensions.}
\label{fig:lattice}
\end{figure}


\section{The weakest classes have no expressiveness}
\label{sec:trivial}
We prove that
\Type{set}{halting}{liberal}{*}-automata have no expressive power,
and the results in Sections~\ref{sec:synchro} and~\ref{sec:exclusive}
will generalize this to \Type{set}{halting}{*}{*}-automata.
Intuitively,
if agents cannot count their neighbors,
and must reach a halting configuration,
then they cannot distinguish any two graphs.
Formally,
a graph property is \emph{trivial} if either every graph satisfies it,
or no graph satisfies it.
We have:

\begin{restatable}{theorem}{ThmTrivial}
  \label{thm:trivial}
  Every \Type{set}{halting}{liberal}{*}-automaton
  recognizes a trivial graph property.
\end{restatable}

\begin{proof}[Proof sketch]
  By Statement~3 of Lemma~\ref{lem:simple},
  it suffices to prove the claim for
  \Type{set}{halting}{liberal}{strong}-automata.
  So let~$A$ be a \Type{set}{halting}{liberal}{strong}-automaton,
  and let~$G$ and~$H$ be two graphs
  (connected and with at least two nodes by convention).
  Assume that~$A$ accepts~$G$ but rejects~$H$.
  By the consistency condition,
  all fair runs of~$A$ on~$G$ accept,
  and all fair runs on~$H$ reject.
  Now let~$\Run^G$ and~$\Run^H$ be any such runs,
  and let $t \in \Naturals$ be a time at which
  all nodes in~$\Run^G$ and~$\Run^H$ have halted.
  We define a new graph~$K$ that consists of
  $t$ copies
  $\Set{G_i}_{i \in \Range[1]{t}}$
  and
  $\Set{H_i}_{i \in \Range[1]{t}}$
  of~$G$ and~$H$,
  with additional edges defined as follows.
  For each node~$w^X$ of the original graph $X \in \Set{G, H}$,
  we denote its copy in~$X_i$ by~$w^X_i$,
  where $i \in \Range[1]{t}$.
  Let~$u^G$ and~$v^G$ be two \emph{adjacent} nodes of~$G$,
  and~$u^H$ and~$v^H$ be two \emph{adjacent} nodes of~$H$.
  We add the connecting edges
  $\Set{u^X_i, v^X_{i+1}}$
  for all $i \in \Range[1]{t}*$ and $X \in \Set{G, H}$,
  as well as the edge
  $\Set{u^G_t, u^H_t}$.
  This is illustrated in Figure~\ref{fig:chain-construction}.

  \begin{figure}[htb]
    \centering
    \begin{tikzpicture}[semithick,auto,on grid]
  \tikzstyle{vertex}=[draw,circle,inner sep=0ex,minimum size=4ex]
  \tikzstyle{vertexG}=[vertex,fill=black!20]
  \tikzstyle{vertexH}=[vertex,fill=white]
  \tikzstyle{phantom}=[]
  \tikzstyle{graph}=[draw,semithick,inner ysep=-0.5ex,
                     cloud,cloud ignores aspect,cloud puff arc=80]
  \tikzstyle{graphG}=[graph,label={[xshift=-2ex]center:#1},
                      cloud puffs=10]
  \tikzstyle{graphH}=[graph,label={[xshift=+2ex]center:#1},fill=black!20,
                      cloud puffs=14]
  \def\dv{7ex}
  \def\dh{12.5ex}

  \node[vertexG] (uG1)                    {$u^G_1$};
  \node[vertexG] (vG1) [below=\dv of uG1] {$v^G_1$};
  \node[vertexG] (uG2) [right=\dh of uG1] {$u^G_2$};
  \node[vertexG] (vG2) [below=\dv of uG2] {$v^G_2$};
  \node[phantom] (pG)  [below right=0.5*\dv and 0.75*\dh of uG2] {$\dots$};
  \node[vertexG] (uGt) [right=1.5*\dh of uG2] {$u^G_t$};
  \node[vertexG] (vGt) [below=\dv of uGt] {$v^G_t$};
  \node[vertexH] (uHt) [right=\dh of uGt] {$u^H_t$};
  \node[vertexH] (vHt) [below=\dv of uHt] {$v^H_t$};
  \node[phantom] (pH)  [below right=0.5*\dv and 0.75*\dh of uHt] {$\dots$};
  \node[vertexH] (uH2) [right=1.5*\dh of uHt] {$u^H_2$};
  \node[vertexH] (vH2) [below=\dv of uH2] {$v^H_2$};
  \node[vertexH] (uH1) [right=\dh of uH2] {$u^H_1$};
  \node[vertexH] (vH1) [below=\dv of uH1] {$v^H_1$};
  \begin{scope}[on background layer]
    \node[graphG=$G_1$,fit=(uG1) (vG1)] (G1) {};
    \node[graphG=$G_2$,fit=(uG2) (vG2)] (G2) {};
    \node[graphG=$G_t$,fit=(uGt) (vGt)] (Gt) {};
    \node[graphH=$H_t$,fit=(uHt) (vHt)] (Ht) {};
    \node[graphH=$H_2$,fit=(uH2) (vH2)] (H2) {};
    \node[graphH=$H_1$,fit=(uH1) (vH1)] (H1) {};
  \end{scope}

  \draw
    (uG1) edge (vG1)
          edge (vG2)
    (uG2) edge (vG2)
          edge (pG)
    (pG)  edge (vGt)
    (uGt) edge (vGt)
          edge (uHt)
    (uHt) edge (vHt)
    (pH)  edge (vHt)
    (uH2) edge (vH2)
          edge (pH)
    (uH1) edge (vH1)
          edge (vH2)
    ;
\end{tikzpicture}

    \caption{
      Graph~$K$ used in the proof of Theorem~\ref{thm:trivial}.
    }
    \label{fig:chain-construction}
  \end{figure}
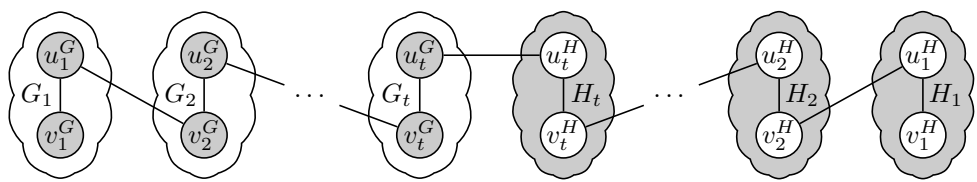

  We show that there is a fair run~$\Run$ of~$A$ on~$K$
  that neither accepts nor rejects.
  It follows that $A$ does not satisfy the consistency condition,
  contradicting the hypothesis.
  Since $A$ is a non-counting automaton,
  initially every node~$w^X_i$ except for~$u^G_t$ and~$u^H_t$
  “sees” the same neighborhood as
  the corresponding node~$w^X$ in the original graph~$X$.
  Only the two nodes~$u^G_t$ and~$u^H_t$ may have a different neighborhoods
  than~$u^G$ and~$u^H$,
  and this might affect their behavior starting at time~$1$.
  Their different behavior can be propagated to other nodes
  in subsequent rounds,
  but it takes time before it reaches every node.
  We exploit this to construct~$\Run$ in such a way that
  some nodes of~$K$ (those of~$G_1$) reach an accepting state,
  while others (those of~$H_1$) reach a rejecting state.
  Since $A$ is a halting automaton,
  these nodes will never change their state again,
  and so the run is neither accepting nor rejecting.
\end{proof}


\section{Synchronicity can always be simulated}
\label{sec:synchro}
We show that every class with synchronous selection is equivalent to the corresponding class with liberal selection.
Albeit non-trivial, this is easy to prove by a standard technique of distributed computing known as \emph{alpha synchronizer}. (The term was introduced in~\cite{Awerbuch85}, but a similar idea appeared earlier in cellular automata theory~\cite{Nakamura81}.) Given a machine $M=\Tuple{Q, \delta_0, \delta, Y, N}$, we define a machine $\tilde{M}=\Tuple{\tilde{Q}, \tilde{\delta_0}, \tilde{\delta}, \tilde{Y}, \tilde{N}}$ such that for every graph~$G$, the unique synchronous run of $M$ on $G$ accepts (rejects) if{}f every weakly fair run $\Run$ of $\tilde{M}$ on $G$ accepts (rejects). The gadget achieving this is called a ``synchronizer'', because it ensures that the nodes of $G$ behave ``as in the synchronous case'', even when selection is~liberal. 

The set of states of $\tilde{M}$ is $\tilde{Q} := Q \times Q \times \{0,1,2\}$. Given $(q, q', i) \in \tilde{Q}$, we call $q$ the \emph{past $M$-state}, $q'$ the \emph{current $M$-state}, and $i$ the \emph{phase}. The initialization function is given by $\tilde{\delta}_0 (a) := (\delta_0(a), \delta_0(a), 0)$. In order to define the transition function $\tilde{\delta}$, let $v$ be a node in state $(q, q', i)$. If $v$ is selected by the scheduler, its next state is determined as follows: 
\begin{itemize}
\item If at least one neighbor of $v$ is in phase $(i-1) \bmod 3$, then $v$ does not change state.\\
Intuitively, if some neighbor is still one phase behind, then $v$ waits for it to ``catch up''.
\item If every neighbor of $v$ is in phase $i$ or $(i+1) \bmod 3$, then $v$  moves to $(q', q'', (i+1) \bmod 3)$, where $q''$ is defined as follows. Let $N_v$ be the set of neighbors of $v$, and for each $u \in N_v$, let $(q_u, q_u', i_u)$ be the state of $u$. Further, let $q_u'' \DefEq q'_u$ if $i_u = i$, and $q_u'' \DefEq q_u$  if $i_u = (i+1) \bmod 3$, and let
${\cal M}$ be the multiset over $Q$ containing for each $u \in N_v$ a copy of the state $q''_u$. (Loosely speaking, ${\cal M}$ contains the current $M$-states of the neighbors of $v$ that are in the same phase as $v$, and the past $M$-states of the neighbors that are one phase ahead, i.e., the states they had when they were in the same phase as $v$). Let ${\cal M}_\beta$ be given by ${\cal M}_\beta(q) = \min \{ \beta, {\cal M}_\beta(q) \}$. We define $q'' := \delta(q, {\cal M}_\beta)$; loosely speaking, $v$ moves to the state it would move to in $M$ if all its neighbors were in the same phase.
\end{itemize}
Let $\tilde{\Run}$ be any weakly-fair run of $\tilde{M}$ on a graph $G$. Fix a node $v$ of $G$, and extract from $\tilde{\Run}$ the sequence $q'_{10} q'_{11} q'_{12} \, q'_{20} q'_{21} q'_{22} \ldots q'_{i0} q'_{i1} q'_{i2} \ldots$, where $q'_{ij}$ denotes the current $M$-state of $v$ immediately after entering phase $j$ for the $i$-th time. Now, let $\Run$ be the unique synchronous run of $M$ on $G$, and let $q_0' q'_1 q'_2 \ldots$ be the sequence obtained by projecting $\Run$ onto the states of~$v$. It is easy to see that these two sequences coincide.  By the definition of stable acceptance, $\tilde{\Run}$ accepts if{}f $\Run$ accepts, and rejects if{}f $\Run$ rejects.
Using this construction, we obtain:

\begin{restatable}{theorem}{ThmSynchronizer}
  \label{lem:synchronizer}
  For every \Type{*}{*}{synchronous}{*}-automaton there is an equivalent \Type{*}{*}{liberal}{*}-automaton.
\end{restatable}


\section{Exclusivity does not increase expressiveness}
\label{sec:exclusive}
In this section,
we obtain the rather surprising result that
the computational power of a class of automata does not increase
if we restrict its schedulers to interleaving ones
(which guarantees that agents act in mutual exclusion with all other agents).

\subsection{Exclusivity under strong fairness}
\label{subsec:ex-strong}

We start by considering strongly fair models,
i.e., we compare a class of the form
\Type{*}{*}{liberal}{strong}
with the corresponding class \Type{*}{*}{exclusive}{strong}.
On an intuitive level,
their equivalence
might be less surprising than the subsequent result
presented in Section~\ref{ssec:exclusivity-under-weak-fairness}
because strong fairness provides a way to break symmetry,
which can be exploited to simulate exclusivity.
Nevertheless,
neither class trivially subsumes the other,
so we have to prove inclusions in both directions.

\begin{restatable}{theorem}{ThmSimulateLiberalFairnessWithExclusiveFairness}
  \label{thm:simulate-liberal-fairness-with-exclusive-fairness}
  For every \Type{*}{*}{liberal}{strong}-automaton
  there is an equivalent \Type{*}{*}{exclusive}{strong}-automaton.
\end{restatable}
\begin{proof}[Proof sketch]
  Given a \Type{*}{*}{liberal}{strong}-automaton~$A$,
  we construct a \Type{*}{*}{exclusive}{strong}-automaton~$B$ such that
  for all input graphs~$G$,
  every strongly fair run of~$B$ on~$G$
  simulates a strongly fair run of~$A$ on~$G$.
  The difficulty lies in the fact that
  $A$ and~$B$ do not share the same notion of strong fairness
  because they have different selection constraints.
  While $A$'s liberal scheduler guarantees that
  arbitrary sequences of selections will occur infinitely often,
  $B$'s exclusive scheduler can select only one node at a time.
  To simulate $A$'s behavior with~$B$,
  we adapt the synchronizer
  from Section~\ref{sec:synchro}.
  Just like there, nodes keep track of
  their previous and current state in~$A$,
  as well as the current phase number modulo~$3$.
  However,  instead of updating their state in every phase,
  they only do so if an additional \emph{activity flag} is set.
  Thus,
  we can simulate an arbitrary selection~$\Selection$
  by raising the flags
  of exactly those nodes that lie in~$\Selection$.
  The outcome of a phase simulated in this way
  will be the same as if
  all the nodes in~$\Selection$ made a transition simultaneously.
  The main issue is how to set the activity flags in each phase
  in such a way that
  every finite sequence
  $(\Selection_1, \ldots, \Selection_n)$ of selections
  is guaranteed to occur infinitely often.
  We show that this is possible,
  exploiting the fact that $B$'s scheduler is strongly fair.
\end{proof}

\begin{restatable}{theorem}{ThmSimulateExclusivityWithFairness}
  \label{thm:simulate-exclusivity-with-fairness}
  For every \Type{*}{*}{exclusive}{strong}-automaton
  there is an equivalent \Type{*}{*}{liberal}{strong}-automaton.
\end{restatable}
\begin{proof}[Proof sketch]
  First,
  we note that the only way exclusivity could possibly be useful
  is to break symmetry between adjacent nodes.
  This is because for an independent set
  (i.e., a set of pairwise non-adjacent nodes),
  the order of activation is irrelevant:
  whether the scheduler activates them
  all at once or one by one in some arbitrary order,
  the outcome will always be the same.
  Consequently, to simulate a run with exclusivity,
  it suffices to simulate a run
  where no two adjacent nodes are active at the same time.
  We provide a simple protocol that makes use of the strong fairness constraint
  (in an environment with liberal selection)
  to ensure that if a node wants to execute a transition,
  then it will eventually be able to do so
  while all its neighbors remain passive.
\end{proof}

\subsection{Exclusivity under weak fairness}
\label{ssec:exclusivity-under-weak-fairness}

We now show that even in the absence of strong fairness,
the restriction to interleaving schedulers
does not increase expressive power.
At first sight,
this may be quite surprising
because exclusivity inherently breaks symmetry,
whereas an automaton with liberal selection and weak fairness
can always be assumed to run synchronously
and thus be incapable of breaking symmetry.
In fact, it is easy to come up with examples of automata
that exploit exclusivity to ensure termination.

\begin{restatable}{proposition}{RemNonhaltingWithoutExclusivity}
  \label{rem:nonhalting-without-exclusivity}
  For every \Type{*}{*}{liberal}{weak}-automaton,
  there exists a \Type{*}{*}{exclusive}{weak}-automaton
  that recognizes the same graph language
  but makes use of exclusive selection to ensure termination.
  If run synchronously,
  it never terminates
  (and hence it is not a valid \Type{*}{*}{liberal}{weak}-automaton).
\end{restatable}

However,
although the automata described in
Proposition~\ref{rem:nonhalting-without-exclusivity}
make use of exclusivity,
they do not really benefit from it;
they only recognize languages
that can also be recognized by liberal automata.
As we will see in
Theorem~\ref{thm:simulate-exclusive-weak-with-liberal-weak},
this observation can be generalized
to arbitrary \Type{*}{*}{exclusive}{weak}-automata.
Intuitively,
since exclusivity does not add any expressive power,
it can in a certain sense be simulated
without needing to break symmetry.

The proof of
Theorem~\ref{thm:simulate-exclusive-weak-with-liberal-weak}
is based on the notion of Kronecker cover.
The \intro{Kronecker cover}
(also known as \intro{bipartite double cover})
of a graph
$G = \Tuple{V, E, \lambda}$
is the bipartite graph
$G' = \Tuple{V', E', \lambda'}$
where
$V' = V \times \Set{0, 1}$, $ E' = \bigcup_{\Set{u,v} \in E}
  \Set{ \,
    \Set{\Tuple{u,0}, \Tuple{v,1}}, \,
    \Set{\Tuple{u,1}, \Tuple{v,0}} \,
  }$,
and
$\lambda'(\Tuple{v, i}) = \lambda(v)$
for all $\Tuple{v, i} \in V'$.
An example is provided in Figure~\ref{fig:butterfly-unfolding}.

\begin{figure}[h]
  \centering
  \begin{subfigure}[b]{0.3\textwidth}
    \begin{tikzpicture}[
    semithick,
    nodes={draw,circle,
           inner sep=0ex,outer sep=-0.1ex,
           minimum size=3ex,
           anchor=center}
  ]
  \node[fill=black,text=white] (u) {$u$};
  \node[below=5ex of u,fill=black!30] (w) {$w$};
  \node[below=5ex of w] (x) {$x$};
  \node (v) at ([xshift=-8ex]$(u)!0.5!(w)$) {$v$};
  \draw (u) -- (v)
        (v) -- (w)
        (w) -- (u)
        (w) -- (x);
\end{tikzpicture}

  \end{subfigure}
  \begin{subfigure}[b]{0.3\textwidth}
    \begin{tikzpicture}[
    semithick,
    nodes={draw,rounded corners=1.3ex,
           inner sep=0ex,outer sep=-0.1ex,
           minimum width=4.8ex,minimum height=3ex,
           anchor=center}
  ]
  \node[fill=black,text=white] (u) {$u,0$};
  \node[below=5ex of u,fill=black!30] (w) {$w,0$};
  \node[below=5ex of w] (x) {$x,0$};
  \node (v) at ([xshift=-8ex]$(u)!0.5!(w)$) {$v,0$};
  \node[right=12ex of u,fill=black,text=white] (u') {$u,1$};
  \node[below=5ex of u',fill=black!30] (w') {$w,1$};
  \node[below=5ex of w'] (x') {$x,1$};
  \node (v') at ([xshift=8ex]$(u')!0.5!(w')$) {$v,1$};
  \draw (u)  -- (v')
        (v)  -- (w')
        (w)  -- (u')
        (w)  -- (x')
        (u') -- (v)
        (v') -- (w)
        (w') -- (u)
        (w') -- (x);
\end{tikzpicture}

  \end{subfigure}
  \caption{
    A graph (on the left) and its Kronecker cover (on the right).
  }
  \label{fig:butterfly-unfolding}
\end{figure}
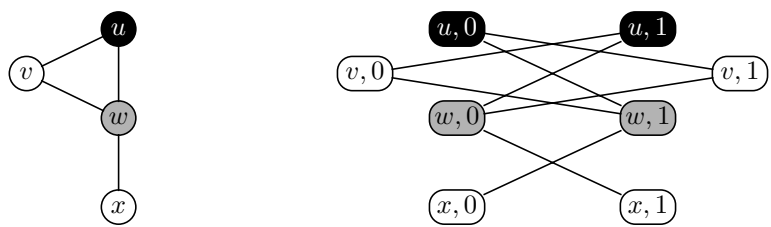

The Kronecker cover in Figure~\ref{fig:butterfly-unfolding}
is connected because the nodes in
$\Set{u, v, w} \times \Set{0, 1}$
form a cycle.
The following lemma generalizes this observation.

\begin{restatable}{lemma}{LemButterflyUnfoldingConnected}
  \label{lem:butterfly-unfolding-connected}
  The Kronecker cover of a connected graph~$G$ is connected
  if and only if
  $G$ contains a cycle of odd length,
  (i.e., if and only if $G$ is non-bipartite).
\end{restatable}

If a Kronecker cover is connected,
then it constitutes a legal input for a distributed automaton.
The next key lemma shows that, in this case,
a weakly fair automaton cannot even distinguish between
a graph and its Kronecker cover.

\begin{restatable}{lemma}{LemButterflyUnfoldingInvariance}
 \label{lem:butterfly-unfolding-invariance}
  For every \Type{*}{*}{*}{weak}-auto\-ma\-ton $A$
  with input alphabet $\Alphabet$
  and every non-bipartite $\Alphabet$-labeled graph~$G$,
  $A$ accepts~$G$
  if and only if
  it accepts the Kronecker cover of~$G$.
\end{restatable}

We can now prove the main technical result of this section:

\begin{restatable}{theorem}{ThmSimulateExclusiveWeakWithLiberalWeak}
  \label{thm:simulate-exclusive-weak-with-liberal-weak}
  For every \Type{*}{*}{exclusive}{weak}-automaton
  there is an equivalent \Type{*}{*}{liberal}{weak}-automaton.
\end{restatable}

\begin{proof}[Proof sketch]
  Given a \Type{*}{*}{exclusive}{weak}-automaton~$A$,
  we construct an equivalent \Type{*}{*}{synchronous}{weak}-automaton~$B$
  (i.e., a synchronous automaton).
  This is sufficient to prove the claim,
  because we know from Theorem~\ref{lem:synchronizer} that
  $B$ can always be simulated
  by a \Type{*}{*}{liberal}{weak}-automaton using a synchronizer.

  Let $G$ be an input graph for $A$.
  If we were guaranteed that
  the labels of $G$ define a proper vertex coloring
  (i.e., edges connect nodes of different colors),
  then the task would be straightforward.
  Indeed,
  since each color of a proper coloring represents an independent set,
  $B$~could simply operate in cyclically repeating phases,
  each one activating precisely the nodes of one of the colors.
  As explained in the proof of
  Theorem~\ref{thm:simulate-exclusivity-with-fairness},
  such a run is equivalent to a run of an exclusive scheduler
  that activates the nodes of each independent set one by one
  (in some arbitrary order).

  This approach can be adapted to bipartite graphs
  because a bipartite graph has exactly two possible 2-colorings.
  However,
  computing one of the two 2-colorings would require to break symmetry,
  which a \Type{*}{*}{synchronous}{weak}-automaton cannot do.
  So instead,
  the states of automaton~$B$ have two components,
  one corresponding to each coloring,
  and nodes update both components when they are activated.

  Using these ideas,
  we construct~$B$ in such a way that
  it recognizes the same bipartite graphs as~$A$.
  Then we use Lemmas~\ref{lem:butterfly-unfolding-connected}
  and \ref{lem:butterfly-unfolding-invariance} to prove that $L(A)=L(B)$.
  Indeed, if~$G$ is not bipartite,
  then by Lemma~\ref{lem:butterfly-unfolding-connected},
  its Kronecker cover~$G'$ is connected
  and therefore constitutes a legal input for a distributed automaton.
  By Lemma~\ref{lem:butterfly-unfolding-invariance},
  $B$~accepts~$G$ if and only if it accepts~$G'$.
  Since Kronecker covers are bipartite by definition,
  we know from the above discussion that
  $B$~accepts~$G'$ if and only if $A$~accepts~$G'$.
  Finally,
  again by Lemma~\ref{lem:butterfly-unfolding-invariance},
  $A$~accepts~$G'$ if and only if it accepts~$G$.
  From this chain of equivalences,
  we can conclude that
  $G$ is accepted by~$B$ if and only if it is accepted by~$A$.
\end{proof}


\section{Separations}
\label{sec:separations}
In Sections \ref{sec:trivial}, \ref{sec:synchro} and \ref{sec:exclusive} we have shown that the classes of graph languages in Figure \ref{fig:lattice} collapse to \emph{at most} the seven classes shown on the left of Figure \ref{fig:lattice-CAF}. In this section we show that the seven classes are all different. For this we 
examine four graph languages, and determine which classes are expressive enough to recognize them:
\newcommand{\Black}{\mathcal{B}}
\newcommand{\Stars}{\mathcal{S}}
\newcommand{\Triangle}{\mathcal{C}_3}
\renewcommand{\triangle}{\textsc{C}_3}
\newcommand{\hexagon}{\textsc{C}_6}
\newcommand{\EvStars}{\mathcal{S}_{\mathit{even}}}
\begin{itemize}
\item $\Black$: The language of graphs with set of labels $\{ \textit{black}, \textit{white} \}$ having at least one black node.
\item $\Stars$: The language of star graphs, i.e., the set of all connected, unlabeled graphs in which one node (the \emph{center}) 
has degree at least 2, and all others (the \emph{leaves}) have degree~1. 
\item $ \Triangle$: The language containing one single graph, namely the cycle $\triangle$ with three nodes labeled by $0$, $1$, and $2$, respectively.
\item $\EvStars$: The language of even stars, i.e., the graphs of $\Stars$ with an even number of leaves. 
\end{itemize}
The results are summarized on the right of Figure \ref{fig:lattice-CAF}.

\begin{figure}[htb]
  \centering
  \begin{subfigure}[c]{0.35\textwidth}
    \begin{tikzpicture}[semithick,>=stealth',shorten >=0.5pt,auto,on grid]
  \tikzstyle{class}=[draw,rounded corners=1.2ex,
                     inner sep=0ex,minimum height=3ex,minimum width=6ex]
  \tikzstyle{D}  =[draw=drawD, fill=fillD]
  \tikzstyle{A}  =[draw=drawA, fill=fillA]
  \tikzstyle{F}  =[draw=drawF, fill=fillF]
  \tikzstyle{DA} =[draw=drawDA,fill=fillDA]
  \tikzstyle{DF} =[draw=drawDF,fill=fillDF]
  \tikzstyle{AF} =[draw=drawAF,fill=fillAF]
  \tikzstyle{DAF}=[draw=drawDAF,shade,top color=topDAF,bottom color=botDAF]
  \def\dv{8ex}
  \def\dh{10ex}

  \node [class,F]   (0000)                                   {\Type{set}{halting}{liberal}{weak}};
  \node [class,D]   (1000) [above left=\dv and \dh of 0000]  {\Type{multiset}{halting}{liberal}{weak}};
  \node [class,A]   (0010) [above right=\dv and \dh of 0000] {\Type{set}{stabilizing}{liberal}{weak}};
  \node [class,DF]  (1100) [above=\dv of 1000]               {\Type{multiset}{halting}{liberal}{strong}};
  \node [class,DA]  (1010) [above=2*\dv of 0000]             {\Type{multiset}{stabilizing}{liberal}{weak}};
  \node [class,AF]  (0110) [above=\dv of 0010]               {\Type{set}{stabilizing}{liberal}{strong}};
  \node [class,DAF] (1110) [above=\dv of 1010]               {\Type{multiset}{stabilizing}{liberal}{strong}};

  \draw[->,detection]
    (0000) edge (1000)
    (0010) edge (1010)
    (0110) edge (1110)
    ;

  \draw[->,acceptance]
    (0000) edge (0010)
    (1000) edge (1010)
    (1100) edge (1110)
    ;

  \draw[->,fairness]
    (1000) edge (1100)
    (0010) edge (0110)
    (1010) edge (1110)
    ;
\end{tikzpicture}

  \end{subfigure}
  \qquad
  \begin{subfigure}[c]{0.35\textwidth}
    \renewcommand{\arraystretch}{1.05}
    \begin{tabular}{rcccc}
      Class & $\;\Black\;$ & $\;\Stars\;$ & $\Triangle$ & $\EvStars$ \\
      \midrule
      \Type{multiset}{stabilizing}{liberal}{strong} & \cmark & \cmark & \cmark & \cmark \\
      \Type{multiset}{halting}{liberal}{strong}     & \xmark & \cmark & \xmark & \cmark \\
      \Type{multiset}{stabilizing}{liberal}{weak}   & \cmark & \cmark & \xmark & \xmark \\
      \Type{set}{stabilizing}{liberal}{strong}      & \cmark & \cmark & \cmark & \xmark \\
      \Type{multiset}{halting}{liberal}{weak}       & \xmark & \cmark & \xmark & \xmark \\
      \Type{set}{stabilizing}{liberal}{weak}        & \cmark & \xmark & \xmark & \xmark \\
      \Type{set}{halting}{liberal}{weak}            & \xmark & \xmark & \xmark & \xmark \\
    \end{tabular}
  \end{subfigure}
  \caption{
    On the left,
    quotient of the classification of Figure~\ref{fig:lattice}.
    On the right,
    four graph languages,
    and the automata models capable of recognizing them.
  }
  \label{fig:lattice-CAF}
\end{figure}

\paragraph*{Recognizing properties of labeled graphs: the language $\mathcal{B}$}
The main difference between the two types of acceptance is that
halting automata cannot recognize properties that require nodes to wait an unlimited amount of time for some information that may never arrive, while even the simplest class of automata accepting by stable consensus can recognize some of those properties,
such as~$\Black$.

\begin{proposition}
\label{prop:ASf-non-trivial}
 $\mathcal{B}$ is recognizable by a \Type{set}{stabilizing}{liberal}{weak}-automaton, but not by any \Type{*}{halting}{*}{*}-automaton.
\end{proposition}
\begin{proof}[Proof sketch]
The \Type{set}{stabilizing}{liberal}{weak}-automaton has two states, called \textit{black} and \textit{white}. The initial state of a node is given by its label. Black nodes remain always black, and white nodes with a black neighbor become black. Since graphs are connected by assumption, if a graph contains some black node then eventually all nodes are black, otherwise all nodes stay white. 

For the second part,
one can show that \Type{multiset}{halting}{liberal}{strong}-automata
cannot distinguish between
an entirely white cycle and a sufficiently long path graph
whose nodes are all white except for two black nodes at the endpoints.
(The argument is similar to the proof of Theorem~\ref{thm:trivial}.)
\end{proof}

\paragraph*{Recognizing properties of unlabeled graphs: the language $\Stars$}
We show in Proposition \ref{prop:ASf-trivial-unlabeled} that \Type{set}{stabilizing}{liberal}{weak}-automata cannot recognize any non-trivial property of \emph{unlabeled} graphs (which we identify with the labeled graphs whose nodes all carry the same label). That is, while \Type{set}{stabilizing}{liberal}{weak}-automata can recognize properties of the labeling of a graph, they cannot recognize any non-trivial property of its \emph{structure}. Then we show in Proposition~\ref{prop:AsF-recog-stars} that the strong fairness of \Type{set}{stabilizing}{liberal}{strong}-automata allows them to recognize $\Stars$.

\begin{proposition}
\label{prop:ASf-trivial-unlabeled}
\Type{set}{stabilizing}{liberal}{weak}-automata can only recognize trivial properties  of \emph{unlabeled} graphs.
In particular, \(\Stars\) is not recognizable by a \Type{set}{stabilizing}{liberal}{weak}-automaton.
\end{proposition}
\begin{proof}
Let $A$ be a \Type{set}{stabilizing}{liberal}{weak}-automaton, and let
$\Run=(\Configuration_0, \Configuration_1, \ldots)$ be the synchronous run of $A$ on an unlabeled graph $G=(V, E)$, i.e., the run scheduled by $V^\omega$. We show that  $A$ either accepts all unlabeled graphs, or rejects all unlabeled graphs. Since $V^\omega$ is a weakly fair schedule, $\Run$ is a fair run, and so by the consistency condition $A$ accepts $G$ if{}f $\Run$ is accepting. Since $G$ is unlabeled, in $\Configuration_0$ every node of $G$ is in the same state $q_0$, which is independent of $G$. Moreover, since~$\Run$ is synchronous and $A$ is non-counting, in each configuration~$\Configuration_i$ every node of $G$ is in the same state $q_i$, which is also independent of $G$. So the states visited by~$\Run$ are independent of $G$, and so $A$ either accepts all unlabeled graphs, or rejects all unlabeled~graphs. 
\end{proof}

\begin{restatable}{proposition}{PropAsFRecogStars}
\label{prop:AsF-recog-stars}
$\Stars$ is recognizable by a \Type{set}{stabilizing}{liberal}{strong}-automaton and by a \Type{multiset}{halting}{liberal}{weak}-automaton.
\end{restatable}
\begin{proof}[Proof sketch] We give a \Type{set}{stabilizing}{liberal}{strong}-automaton that recognizes $\Stars$. 
The states of the automaton are pairs $(d, c)$, where $d \in \{\textit{leaf}, \textit{center}, \textit{unknown}, \textit{neither}\}$ is the \emph{estimate} of $v$, and $c \in \{ 0,1 \}$ is its \emph{color}. Every time a node is selected it flips its color. When a node with estimate \emph{unknown} sees two neighbors with different colors, it switches to \textit{center}, and if from then on it sees a neighbor with estimate \textit{center}, it moves to \textit{neither}. Strong fairness is crucial for correctness: by Lemma~\ref{lem:strongfairness}, it ensures that a node that is not a leaf will eventually be selected in a configuration in which at least two of its neighbors have different colors. 

Now we give a \Type{multiset}{halting}{liberal}{weak}-automaton with $\beta=2$ that recognizes $\Stars$. Since $\beta=2$, a node can determine for each state~$q$ if it has $0$, $1$, or at least~$2$ neighbors in $q$. The automaton's states are $\{ \textit{init}, \textit{leaf}, \textit{non-leaf}, \textit{accept}, \textit{reject} \}$. Initially all nodes are in state $\textit{init}$. The nodes update their estimates depending on the number of neighbors (0, 1, or at least 2) in each state. 
\end{proof}

\paragraph*{Symmetry breaking: the language $\Triangle$}

We show that the language $\Triangle$ requires both acceptance by stable consensus and strong fairness to be recognizable. Intuitively, both of them are required to distinguish $C_3$ from arbitrarily long cycles that repeat the labeling of~$C_3$ cyclically.  

\begin{restatable}{proposition}{PropCthreeSeparation}
\label{prop:Cthree-separation}
$\Triangle$ is recognizable by a \Type{set}{stabilizing}{liberal}{strong}-automaton, but neither by \Type{multiset}{stabilizing}{*}{weak}-automata nor by \Type{multiset}{halting}{*}{strong}-automata.
\end{restatable}

\begin{proof}[Proof sketch]
  Our \Type{set}{stabilizing}{liberal}{strong}-automaton for $\Triangle$
  checks two conditions:
  first,
  that the input graph is a cycle with cyclic labeling ${0{-}1{-}2}$,
  and second,
  that it contains exactly one node labeled by~$2$
  (which implies that the cycle has length~$3$).
  For both conditions,
  we use a similar trick as in Proposition~\ref{prop:AsF-recog-stars},
  relying on acceptance by stable consensus and strong fairness
  to eventually break symmetry between otherwise indistinguishable nodes.
  To verify the second condition,
  each node labeled by~$2$ successively sends signals
  in both directions through the cycle,
  and checks that those signals always come back from the expected direction.

  For the second part of the claim,
  we show that
  \Type{multiset}{stabilizing}{*}{weak}- and
  \Type{multiset}{halting}{*}{strong}-automata
  cannot distinguish~$\triangle$ from~$\hexagon$,
  the hexagon whose nodes are labeled by ${0{-}1{-}2{-}0{-}1{-}2}$
  (and back to~$0$).
  To do so,
  given a fair run~$\Run_3$ of such an automaton on~$\triangle$,
  we construct a fair run~$\Run_6$ on~$\hexagon$
  that “duplicates” the behavior of~$\Run_3$.
  In the case of \Type{multiset}{halting}{*}{strong}-automata,
  this duplication is performed only until
  $\Run_3$ has reached a halting configuration
  (because otherwise $\Run_6$ would violate the strong fairness constraint).
\end{proof}

\paragraph*{Counting neighbors modulo a number: the language $\EvStars$}

Since counting automata can only count up to a threshold~$\beta$, no node can directly observe that
it has an even number of neighbors. This makes the language $\EvStars$ rather difficult to recognize. 
We now show that the combination of counting and strong fairness can do the job. 
The proof also provides a good example where exclusivity helps to design an algorithm.

\begin{restatable}{proposition}{PropCasFRecogStars}
\label{prop:CasF-recog-stars}
$\EvStars$ is recognizable by a \Type{multiset}{halting}{liberal}{strong}-automaton.
\end{restatable}
\begin{proof}[Proof sketch]
In Proposition \ref{prop:AsF-recog-stars} we have exhibited a \Type{multiset}{halting}{liberal}{weak}-automaton $A$ recognizing $\Stars$. 
We now give a  \Type{multiset}{halting}{exclusive}{strong}-automaton $B$ that uses counting, exclusivity, and strong fairness to further decide if the number of leaves is even. Loosely speaking, $B$ first executes $A$; if $A$ rejects, then $B$ rejects, because the graph is not even  a star. If $A$ accepts, then $B$ enters a new phase during which it counts the number of leaves modulo 2. 
By Theorem~\ref{thm:simulate-exclusivity-with-fairness},
$B$ is equivalent to a \Type{multiset}{halting}{liberal}{strong}-automaton.

We can assume that when  $A$ accepts, all nodes are labeled with either \textit{leaf} or
\textit{center} (the unique non-leaf). We give an informal description of $B$. Leaves can be in states \textit{visible}, \textit{invisible}, \textit{dead}, \textit{even}, or \textit{odd}. While leaves have not been counted by the center, they alternate between the states \textit{visible} and \textit{invisible}. The center only increments its modulo-$2$ counter if exactly one leaf is \textit{visible}. After a leaf is counted, it moves to \textit{dead}. When all leaves become dead, i.e., when they have all been counted, the center decides whether to accept or reject; the leaves read the decision from the counter, and move to  \textit{even} or \textit{odd} accordingly. 
\end{proof}

The next two results show that recognizing $\EvStars$ needs both counting and strong fairness.

\begin{proposition}
\label{prop:CA*f-cannot-even-stars}
$\EvStars$ is not recognizable by \Type{multiset}{stabilizing}{*}{weak}-automata.
\end{proposition}
\begin{proof}
We show that for every \Type{multiset}{stabilizing}{*}{weak}-automaton $A$ there exist
stars $G$ and $G'$ such that exactly one of $G$ and $G'$ belongs to \(S_{even}\), but $A$
either accepts both of them or rejects both of them. Let $\beta \geq 1$ be $A$'s counting bound, and let $G$ and $G'$ be the stars with $\beta+1$ and $\beta+2$ leaves,
respectively.
Now consider the synchronous runs~$\Run$ and~$\Run'$ of~$A$ on~$G$ and~$G'$.
By symmetry,
and since the number of leaves exceeds~$\beta$ in both~$G$ and~$G'$,
at every time $t \in \Naturals$,
the center is in the same state in~$\Run$ and~$\Run'$,
and likewise all leaves are in the same state.
So the sequences of states visited by the center and the leaves
are the same in both~$\Run$ and~$\Run'$,
and therefore $\Run$ is accepting if{}f $\Run'$ is accepting.
\end{proof}

\begin{restatable}{proposition}{PropCASFCannotEvenStars}
\label{prop:CA*F-cannot-even-stars}
$\EvStars$ is not recognizable by \Type{set}{stabilizing}{*}{strong}-automata.
\end{restatable}
\begin{proof}[Proof sketch]
Given a \Type{set}{stabilizing}{*}{strong}-automaton~$A$, the proof identifies an even number~$n$, depending on~$A$, such that
if $A$ accepts the star with $n$ leaves, then it cannot reject the star with $n+1$ leaves. The proof is involved, and can be found in the Appendix.
\end{proof}


\section{Expressive power}
\label{sec:pprotocols}
As a first application of our results,
we investigate the expressivity of our models
for graph languages that depend only on the labeling function of a graph,
and not on its topology.

Given a $\Alphabet$-labeled graph $G=(V,E,\lambda)$,
where $\Alphabet = \{ \ell_1, \ldots, \ell_k\}$,
let $\#_G \colon \Alphabet \to \Naturals$ be the mapping
that assigns to each label $\ell$
the number $\#_G(\ell)$ of nodes of $V$ such that $\lambda(v) = \ell$.
A language is \emph{Presburger-definable}
if there is a formula $\varphi(x_1, \ldots, x_k)$ of Presburger arithmetic
such that a $\Alphabet$-labeled graph $G$ belongs to the language
if and only if $\varphi(\#_G(\ell_1), \ldots, \#_G(\ell_k))$ holds.
An example of such a language is~$\Black$,
the set of graphs that contain a black node.

We show that \Type{multiset}{stabilizing}{liberal}{strong}-automata
recognize all Presburger languages,
but none of the other six classes do.
The negative part of the result follows easily from the table in
Figure~\ref{fig:lattice-CAF}.

\begin{restatable}{proposition}{PropNotCoverPresburger}
  \label{prop:not-cover-presburger}
  There exist Presburger-definable languages
  that are not recognizable by
  \Type{set}{*}{*}{*}-,
  \Type{*}{halting}{*}{*}-, or
  \Type{*}{*}{*}{weak}-automata.
\end{restatable}

\begin{proof}
  By Proposition~\ref{prop:ASf-non-trivial},
  \Type{*}{halting}{*}{*}-automata cannot recognize the language~$\Black$,
  which is Presburger-definable.
  Furthermore,
  by Propositions~\ref{prop:AsF-recog-stars},
  \ref{prop:CA*f-cannot-even-stars}
  and~\ref{prop:CA*F-cannot-even-stars},
  \Type{set}{stabilizing}{*}{strong}- and
  \Type{multiset}{stabilizing}{*}{weak}-automata
  can recognize the language~$\Stars$ of star graphs
  but not the language~$\EvStars$ of stars with an even number of leaves.
  This implies that
  \Type{set}{stabilizing}{*}{strong}- and
  \Type{multiset}{stabilizing}{*}{weak}-automata
  cannot recognize
  the Presburger-definable language of graphs with an odd number of nodes,
  because
  the intersection of this language with~$\Stars$ is equal to~$\EvStars$,
  and
  languages recognizable by distributed automata are closed under intersection
  (by a standard product construction).
\end{proof}

For the positive part, we proceed in three steps:
First, following \cite{angluin2005stably} and Section 5 of~\cite{AADFP06}, we introduce \emph{graph population protocols}, a graph variant of the well-known population protocol model introduced in \cite{AngluinADFP04,AADFP06}. Then we recall a result of  \cite{AADFP06} showing that graph population protocols recognize all Presburger-definable languages.
Finally, we show that every graph population protocol can be simulated by a 
\Type{multiset}{stabilizing}{liberal}{strong}-automaton.

Our definition of graph population protocols is equivalent
to that of \cite{angluin2005stably,AADFP06},
but reuses the notation of Section~\ref{sec:prelim} as far as possible.
A \emph{graph population protocol}
$\varPi = \Tuple{Q, \delta_0, \delta, Y, N}$
is defined like a \Type{multiset}{stabilizing}{exclusive}{strong}-automaton
with machine $M = \varPi$,
except for the following differences:
\begin{itemize}
\item The transition function is of the form
  $\delta \colon Q^2 \to Q^2$.
\item A selection of a graph $G = \Tuple{V, E, \lambda}$
  is an ordered pair
  $\Selection = \Tuple{u, v} \in V^2$
  of \emph{adjacent} nodes
  (instead of a singleton $\Set{u} \subseteq V$),
  and the selection constraint on~$G$ is
  $\SetBuilder{\Tuple{u, v}}{\Set{u, v} \in E}$.
\item  $\Configuration_t(v)$ is defined inductively as follows,
  for $t \in \Naturals$ and $v \in V$:
  \begin{equation*}
    \Configuration_0(v) = \delta_0(\lambda(v))
    \quad \text{and} \quad
    \Configuration_{t+1}(v) =
    \begin{cases*}
      \delta\bigTuple{\Configuration_t(v), \Configuration_t(u)}_{\operatorname{fst}}
      & if $\Selection_t = \Tuple{v, u}$ for some~$u$, \\
      \delta\bigTuple{\Configuration_t(u), \Configuration_t(v)}_{\operatorname{snd}}
      & if $\Selection_t = \Tuple{u, v}$ for some~$u$, \\
      \Configuration_t(v)
      & otherwise,
    \end{cases*}
  \end{equation*}
  where $P_{\operatorname{fst}}$ and $P_{\operatorname{snd}}$
  denote the first and second component of a pair~$P$.
\end{itemize}
So, intuitively,
the scheduler selects two adjacent nodes,
which update their states according to~$\delta$.
The definitions of all other relevant notions remain the same.
This holds in particular for
acceptance by stable consensus and strong fairness
(which are baked into the model),
and the consistency condition.
Standard population protocols correspond to
graph population protocols on complete graphs,
where every pair of distinct nodes is connected by an edge.

It is shown in \cite{AADFP06} that standard population protocols recognize all Presburger-definable languages. Further, Theorem 7 of \cite{AADFP06} shows that every language recognized by population protocols is also recognized by graph population protocols. Loosely speaking, given a population protocol, one constructs the protocol on graphs in which, when an edge of the graph is selected, either the two nodes connected by it interact as in the population protocol, or they swap their states. By strong fairness, the states of the nodes can ``move around the graph'', and any pair of states eventually interacts infinitely often. The choice between interacting or swapping is nondeterministic, but it can be simulated by deterministic transitions (see \cite{AADFP06}). Therefore, in order to show that  \Type{multiset}{stabilizing}{*}{strong}-automata recognize all Presburger-definable languages, it suffices to simulate graph population protocols with distributed automata. As in the proof of Proposition~\ref{prop:CasF-recog-stars},
we make use of exclusivity to simplify the construction.

\begin{restatable}{proposition}{PropPopulationProtocols}
  \label{prop:population-protocols}
  For every graph population protocol
  there is an equivalent \Type{multiset}{stabilizing}{*}{strong}-automaton.
\end{restatable}

\begin{proof}[Proof sketch]
  We present a simulation
  that runs a population protocol on a distributed automaton.
  To this end,
  the automaton has to simulate a scheduler that selects
  ordered pairs of adjacent nodes instead of arbitrary sets of nodes.
  For any pair $\Tuple{u, v}$ that is selected to perform a transition,
  let us call~$u$ the \emph{initiator} and~$v$ the \emph{responder}
  of the transition.
  By Theorem~\ref{thm:simulate-exclusivity-with-fairness},
  we may assume that
  the automaton's scheduler selects a single node in each step.

  The main idea is as follows:
  When a node $u$ is selected and sees that
  it can become the initiator of a transition,
  it declares its intention to do so by raising the flag~“?”.
  Then $u$ waits until
  some neighbor~$v$ is selected and raises the flag~“!”,
  which signals that $v$ wants to become the responder of a transition.
  If this happens,
  the next time $u$ is selected,
  it computes its new state according to
  the state of $v$ and the transition function of the population protocol,
  but also keeps its old state in memory
  so that $v$ can still see it.
  After that,
  $v$~also updates its state,
  and finally $u$ deletes its old state,
  which completes the transition.
  Throughout this protocol,
  the nodes verify that they have exactly one partner during each transition.
  If this condition is violated,
  they raise the error flag~“$\bot$” and abort their current transition.
\end{proof}

\begin{corollary}
\label{cor:cover-presburger}
\Type{multiset}{stabilizing}{*}{strong}-automata
recognize all Presburger-definable languages.
\end{corollary}


\section{Conclusions}
\label{sec:conclusions}
We have conducted an extensive comparative analysis of the expressive power of weak asynchronous models of distributed computing. Our analysis has reduced the initial ``jungle'' of twenty different models to only seven. 
This reduction in complexity is achieved by Theorems \ref{thm:trivial}, \ref{lem:synchronizer}, \ref{thm:simulate-liberal-fairness-with-exclusive-fairness}, \ref{thm:simulate-exclusivity-with-fairness},
and  \ref{thm:simulate-exclusive-weak-with-liberal-weak}, all of which have a clear and intuitive interpretation.

We have also shown that the seven classes are distinct, and have identified inclusions and non-inclusions between them. However, two inclusions remain open: Are \Type{multiset}{halting}{liberal}{weak} or \Type{multiset}{stabilizing}{liberal}{weak} included in  \Type{set}{stabilizing}{liberal}{strong}? Intuitively, this asks if strong fairness and acceptance by stable consensus can be used to simulate counting.
We can provide a positive answer for graphs of bounded degree (a limitation common in practice), because in this case even \Type{set}{stabilizing}{*}{strong} and \Type{multiset}{stabilizing}{liberal}{strong} coincide.
\begin{restatable}{proposition}{PropAFkSimulatesCAF}
\label{prop:AF-k-simulates-CAF}
For every \Type{multiset}{stabilizing}{*}{strong}-automaton $A$  and every $k \in \Naturals$ there is a  \Type{set}{stabilizing}{*}{strong}-automaton $B$ equivalent to $A$ on graphs of maximum degree $k$.
\end{restatable}
\noindent However, for arbitrary graphs we conjecture that
neither \Type{multiset}{halting}{liberal}{weak}
nor \Type{multiset}{stabilizing}{liberal}{weak}
are included in \Type{set}{stabilizing}{liberal}{strong}.

Finally, we have made a first step towards characterizing the graph languages
recognizable by the different classes, by transferring a characterization for population protocols.

As a last note, observe that our results hold for \emph{decision} problems on \emph{undirected} graphs that can be solved by consensus
in the framework of distributed automata. Several of our constructions
(e.g., those in Theorems~\ref{lem:synchronizer}
and~\ref{thm:simulate-exclusivity-with-fairness})
rely on bidirectional communication,
which is not guaranteed on directed graphs.
Furthermore, exclusive selection leads to higher computational power
for non-decision problems. 
For instance, it can be used to solve the vertex coloring problem
on graphs of bounded degree (by a standard greedy algorithm),
which, for symmetry reasons, is impossible in a model with synchronous selection.


\bibliography{tex/references}

\begin{thebibliography}{10}

\bibitem{angluin2005stably}
Dana Angluin, James Aspnes, Melody Chan, Michael~J Fischer, Hong Jiang, and
  Ren{\'e} Peralta.
\newblock Stably computable properties of network graphs.
\newblock In {\em International Conference on Distributed Computing in Sensor
  Systems}, pages 63--74. Springer, 2005.

\bibitem{AngluinADFP04}
Dana Angluin, James Aspnes, Zo{\"{e}} Diamadi, Michael~J. Fischer, and
  Ren{\'{e}} Peralta.
\newblock Computation in networks of passively mobile finite-state sensors.
\newblock In {\em {PODC}}, pages 290--299. {ACM}, 2004.

\bibitem{AADFP06}
Dana Angluin, James Aspnes, Zo{\"{e}} Diamadi, Michael~J. Fischer, and
  Ren{\'{e}} Peralta.
\newblock Computation in networks of passively mobile finite-state sensors.
\newblock {\em Distributed Computing}, 18(4):235--253, 2006.

\bibitem{Awerbuch85}
Baruch Awerbuch.
\newblock Complexity of network synchronization.
\newblock {\em J. {ACM}}, 32(4):804--823, 1985.
\newblock URL: \url{https://doi.org/10.1145/4221.4227}, \href
  {http://dx.doi.org/10.1145/4221.4227} {\path{doi:10.1145/4221.4227}}.

\bibitem{CK10}
Alejandro Cornejo and Fabian Kuhn.
\newblock Deploying wireless networks with beeps.
\newblock In {\em {DISC}}, volume 6343 of {\em Lecture Notes in Computer
  Science}, pages 148--162. Springer, 2010.

\bibitem{Diestel17}
Reinhard Diestel.
\newblock {\em Graph Theory, 5th Edition}, volume 173 of {\em Graduate texts in
  {M}athematics}.
\newblock Springer, 2017.

\bibitem{EmekW13}
Yuval Emek and Roger Wattenhofer.
\newblock Stone age distributed computing.
\newblock In {\em {PODC}}, pages 137--146. {ACM}, 2013.

\bibitem{Francez}
Nissim Francez.
\newblock {\em Fairness}.
\newblock Texts and Monographs in Computer Science. Springer, 1986.

\bibitem{HJKLLLSV15}
Lauri Hella, Matti J{\"{a}}rvisalo, Antti Kuusisto, Juhana Laurinharju, Tuomo
  Lempi{\"{a}}inen, Kerkko Luosto, Jukka Suomela, and Jonni Virtema.
\newblock Weak models of distributed computing, with connections to modal
  logic.
\newblock {\em Distributed Computing}, 28(1):31--53, 2015.

\bibitem{LehmannPS81}
Daniel Lehmann, Amir Pnueli, and Jonathan Stavi.
\newblock Impartiality, justice and fairness: The ethics of concurrent
  termination.
\newblock In {\em {ICALP}}, volume 115 of {\em Lecture Notes in Computer
  Science}, pages 264--277. Springer, 1981.

\bibitem{Nakamura81}
Katsuhiko Nakamura.
\newblock Synchronous to asynchronous transformation of polyautomata.
\newblock {\em J. Comput. Syst. Sci.}, 23(1):22--37, 1981.
\newblock URL: \url{https://doi.org/10.1016/0022-0000(81)90003-9}, \href
  {http://dx.doi.org/10.1016/0022-0000(81)90003-9}
  {\path{doi:10.1016/0022-0000(81)90003-9}}.

\bibitem{NB15}
Saket Navlakha and Ziv Bar{-}Joseph.
\newblock Distributed information processing in biological and computational
  systems.
\newblock {\em Commun. {ACM}}, 58(1):94--102, 2015.
\newblock \href {http://dx.doi.org/10.1145/2678280}
  {\path{doi:10.1145/2678280}}.

\bibitem{Reiter17}
Fabian Reiter.
\newblock Asynchronous distributed automata: {A} characterization of the modal
  mu-fragment.
\newblock In {\em {ICALP}}, volume~80 of {\em LIPIcs}, pages 100:1--100:14.
  Schloss Dagstuhl - Leibniz-Zentrum f{\"{u}}r Informatik, 2017.

\bibitem{SoloveichikCWB08}
David Soloveichik, Matthew Cook, Erik Winfree, and Jehoshua Bruck.
\newblock Computation with finite stochastic chemical reaction networks.
\newblock {\em Natural Computing}, 7(4):615--633, 2008.

\end{thebibliography}

\clearpage
\appendix
\section{Appendix}
\label{sec:appendix}
\subsection{Proofs of Section \ref{sec:prelim}}

\LemStrongFairness*

\begin{proof}
  Let
  $\ConfigurationSet[3] = \Set{\Configuration[3]_1, \dots, \Configuration[3]_k}$
  be the set of configurations
  that occur infinitely often in~$\Run$.
  Notice that these configurations
  can all reach each other
  because otherwise they could not occur infinitely often.
  The assumption is that
  $\Configuration[2]_0 \in \ConfigurationSet[3]$.
  We construct a finite sequence~$\Schedule$ of selections permitted by $A$
  such that
  for every $i \in \Range[1]{k}$,
  the sequence of configurations
  visited starting from~$\Configuration[3]_i$ and applying~$\Schedule$
  contains either the subsequence
  $\Tuple{\Configuration[2]_0, \dots, \Configuration[2]_n}$,
  or some configuration
  $\Configuration[3]'_i \notin \ConfigurationSet[3]$.
  This suffices to prove the claim because
  from a certain point on,
  $\Run$ visits only configurations in~$\ConfigurationSet[3]$,
  and by strong fairness
  the schedule fragment~$\Schedule$
  is guaranteed to be chosen infinitely often by the~scheduler.
  Since any configuration
  $\Configuration[3]'_i \notin \ConfigurationSet[3]$
  may only occur finitely often,
  the only possibility is that the subsequence
  $\Tuple{\Configuration[2]_0, \dots, \Configuration[2]_n}$
  occurs infinitely~often.

  It remains to construct a suitable sequence~$\Schedule$.
  We proceed by induction,
  constructing a series of sequences
  $\Schedule_0, \Schedule_1, \dots, \Schedule_k$
  such that for $j \in \Range[0]{k}$,
  the sequence~$\Schedule_j$ satisfies the desired property
  for every $i \in \Range[1]{j}$.
  It then suffices to choose $\Schedule = \Schedule_k$.
  As the base case,
  we set $\Schedule_0 = \varepsilon$ (the empty sequence).
  Now, given~$\Schedule_j$,
  we distinguish two cases in order to construct~$\Schedule_{j+1}$.
  If starting from~$\Configuration[3]_{j+1}$ and applying~$\Schedule_j$
  the automaton visits some configuration
  $\Configuration[3]'_{j+1} \notin \ConfigurationSet[3]$,
  then we simply set $\Schedule_{j+1} = \Schedule_j$.
  Otherwise,
  let~$\Configuration[3]'_{j+1}$ be the final configuration
  reached from~$\Configuration[3]_{j+1}$ by applying~$\Schedule_j$.
  Since $\Configuration[3]'_{j+1} \in \ConfigurationSet[3]$
  and $\Configuration[2]_0 \in \ConfigurationSet[3]$,
  there exists a sequence of selections~$\Schedule'$
  that leads the automaton
  from~$\Configuration[3]'_{j+1}$ to~$\Configuration[2]_0$.
  Therefore,
  if starting from~$\Configuration[3]_{j+1}$,
  the automaton applies the schedule fragment
  $\Schedule_{j+1} =
  \Schedule_j \cdot \Schedule' \cdot \Selection_0 \cdots \Selection_{n-1}$,
  then it traverses a sequence of configurations ending with
  $\Tuple{\Configuration[2]_0, \dots, \Configuration[2]_n}$.
  Moreover,
  since $\Schedule_j$ is a prefix of~$\Schedule_{j+1}$,
  the property already established for~$\Schedule_j$
  with respect to
  $\Configuration[3]_1, \dots, \Configuration[3]_j$
  also holds for~$\Schedule_{j+1}$.
\end{proof}

\LemSimple*

\begin{proof}

\smallskip \noindent\res{1} Non-counting automata are a subclass of counting automata. 

\smallskip \noindent\res{2} Halting automata are a subclass of automata accepting by stable consensus. 

\smallskip \noindent\res{3}
Let $A=(M,\SelectionConstr,\FairnessConstr)$ be a \Type{*}{*}{*}{weak}-automaton, and let $G=\Tuple{V, E, \lambda}$ be a graph. The set $\FairnessConstr(G)$ contains the weakly-fair runs of $\SelectionConstr(G)^\omega$. Now consider $A'=(M,\SelectionConstr,\FairnessConstr')$, where $\FairnessConstr'(G)$ contains the strongly-fair runs of $\SelectionConstr(G)^\omega$. Since the set of permitted selections is the same for~$A$ and~$A'$, we have  $\FairnessConstr'(G) \subseteq \FairnessConstr(G)$. Therefore, since $A$ satisfies the consistency condition, so does~$A'$, and thus $A'$ is a \Type{*}{*}{*}{strong}-automaton with $L(A')=L(A)$. 

\smallskip \noindent\res{4} 
Let $A=(M,\SelectionConstr,\FairnessConstr)$ be a \Type{*}{*}{liberal}{weak}-automaton, and let $G=\Tuple{V, E, \lambda}$ be a graph. We have $\SelectionConstr(G) = 2^V$, and $\FairnessConstr(G)$ contains the weakly-fair runs of $\SelectionConstr(G)^\omega$. Let $\SelectionConstr'(G)= \{ \{v\} \mid v \in V \}$, and let $\FairnessConstr'(G)$ be the weakly-fair runs of $\SelectionConstr'(G)^\omega$. We have $\FairnessConstr'(G) \subseteq \FairnessConstr(G)$. Proceed now as in \res{3} 

\smallskip \noindent\res{5}  
The argument is fully analogous to that of \res{4},
the only difference being that
$\SelectionConstr'(G)= \Set{V}$.

\smallskip \noindent\res{6}   Let $A=(M,\SelectionConstr,\FairnessConstr)$ be a \Type{*}{*}{synchronous}{strong}-automaton.
We have $\SelectionConstr(G) = \{V\}$. Further,  the run $\Run$ scheduled by $V^\omega$ is strongly fair (because $V$ is the only possible selection). So $\FairnessConstr(G) = \{\Run\}$.
Let $A'=(M,\SelectionConstr,\FairnessConstr')$ be the unique  \Type{*}{*}{synchronous}{weak}-automaton with machine $M$. Since the run $\Run$ scheduled by $V^\omega$ is also weakly fair, we have $\FairnessConstr'(G) = \{\Run\} = \FairnessConstr(G)$. It follows that $L(A) = L(A')$.
\end{proof}

\subsection{Proofs of Section \ref{sec:trivial}}

\ThmTrivial*

\begin{proof}
  By Statement~3 of Lemma~\ref{lem:simple},
  it suffices to prove the claim for
  \Type{set}{halting}{liberal}{strong}-automata.
  So let us consider a
  \Type{set}{halting}{liberal}{strong}-automaton~$A$,
  and assume for the sake of contradiction
  that there exist two graphs~$G$ and~$H$
  such that $A$ accepts~$G$ and rejects~$H$.
  Let
  $\Run^G = \Tuple{\Configuration^G_0, \Configuration^G_1, \dots}$
  and
  $\Run^H = \Tuple{\Configuration^H_0, \Configuration^H_1, \dots}$
  be strongly fair runs of~$A$ on~$G$ and~$H$,
  respectively.
  By the consistency condition,
  $\Run^G$ is accepting and $\Run^H$ is rejecting.
  Based on that,
  we will construct a new graph~$K$
  and a strongly fair run~$\Run$ of~$A$ on~$K$
  that is neither accepting nor rejecting.
  This means that~$A$ does not satisfy the consistency condition,
  and therefore does not qualify as a distributed automaton,
  a contradiction.

  We start by constructing~$K$.
  Let $t \in \Naturals$ be a time at which
  all nodes in~$\Run^G$ and~$\Run^H$ have halted
  (i.e., all nodes in~$\Configuration^G_t$ and~$\Configuration^H_t$
  have reached an accepting or rejecting state).
  Our new graph~$K$ consists of
  $t$ copies $\Set{G_i}_{i \in \Range[1]{t}}$ of~$G$
  and
  $t$ copies $\Set{H_i}_{i \in \Range[1]{t}}$ of~$H$,
  which are connected as follows.
  For each node~$w^X$ of the original graph $X \in \Set{G, H}$,
  we denote its copy in~$X_i$ by~$w^X_i$,
  where $i \in \Range[1]{t}$.
  Let~$u^G$ and~$v^G$ be two \emph{adjacent} nodes of~$G$,
  and~$u^H$ and~$v^H$ be two \emph{adjacent} nodes of~$H$.
  (Recall that all graphs are assumed to be connected
  and have at least two nodes.)
  In addition to the edges in each copy~$X_i$,
  graph~$K$ also contains the connecting edges
  $\Set{u^X_i, v^X_{i+1}}$
  for all $i \in \Range[1]{t}*$ and $X \in \Set{G, H}$,
  as well as the edge
  $\Set{u^G_t, u^H_t}$.
  An illustration of this construction is provided in
  Figure~\ref{fig:chain-construction-again}.

  \begin{figure}[htb]
    \centering
    \begin{tikzpicture}[semithick,auto,on grid]
  \tikzstyle{vertex}=[draw,circle,inner sep=0ex,minimum size=4ex]
  \tikzstyle{vertexG}=[vertex,fill=black!20]
  \tikzstyle{vertexH}=[vertex,fill=white]
  \tikzstyle{phantom}=[]
  \tikzstyle{graph}=[draw,semithick,inner ysep=-0.5ex,
                     cloud,cloud ignores aspect,cloud puff arc=80]
  \tikzstyle{graphG}=[graph,label={[xshift=-2ex]center:#1},
                      cloud puffs=10]
  \tikzstyle{graphH}=[graph,label={[xshift=+2ex]center:#1},fill=black!20,
                      cloud puffs=14]
  \def\dv{7ex}
  \def\dh{12.5ex}

  \node[vertexG] (uG1)                    {$u^G_1$};
  \node[vertexG] (vG1) [below=\dv of uG1] {$v^G_1$};
  \node[vertexG] (uG2) [right=\dh of uG1] {$u^G_2$};
  \node[vertexG] (vG2) [below=\dv of uG2] {$v^G_2$};
  \node[phantom] (pG)  [below right=0.5*\dv and 0.75*\dh of uG2] {$\dots$};
  \node[vertexG] (uGt) [right=1.5*\dh of uG2] {$u^G_t$};
  \node[vertexG] (vGt) [below=\dv of uGt] {$v^G_t$};
  \node[vertexH] (uHt) [right=\dh of uGt] {$u^H_t$};
  \node[vertexH] (vHt) [below=\dv of uHt] {$v^H_t$};
  \node[phantom] (pH)  [below right=0.5*\dv and 0.75*\dh of uHt] {$\dots$};
  \node[vertexH] (uH2) [right=1.5*\dh of uHt] {$u^H_2$};
  \node[vertexH] (vH2) [below=\dv of uH2] {$v^H_2$};
  \node[vertexH] (uH1) [right=\dh of uH2] {$u^H_1$};
  \node[vertexH] (vH1) [below=\dv of uH1] {$v^H_1$};
  \begin{scope}[on background layer]
    \node[graphG=$G_1$,fit=(uG1) (vG1)] (G1) {};
    \node[graphG=$G_2$,fit=(uG2) (vG2)] (G2) {};
    \node[graphG=$G_t$,fit=(uGt) (vGt)] (Gt) {};
    \node[graphH=$H_t$,fit=(uHt) (vHt)] (Ht) {};
    \node[graphH=$H_2$,fit=(uH2) (vH2)] (H2) {};
    \node[graphH=$H_1$,fit=(uH1) (vH1)] (H1) {};
  \end{scope}

  \draw
    (uG1) edge (vG1)
          edge (vG2)
    (uG2) edge (vG2)
          edge (pG)
    (pG)  edge (vGt)
    (uGt) edge (vGt)
          edge (uHt)
    (uHt) edge (vHt)
    (pH)  edge (vHt)
    (uH2) edge (vH2)
          edge (pH)
    (uH1) edge (vH1)
          edge (vH2)
    ;
\end{tikzpicture}

    \caption{
      Graph~$K$ used in the proof of Theorem~\ref{thm:trivial}.
    }
    \label{fig:chain-construction-again}
  \end{figure}
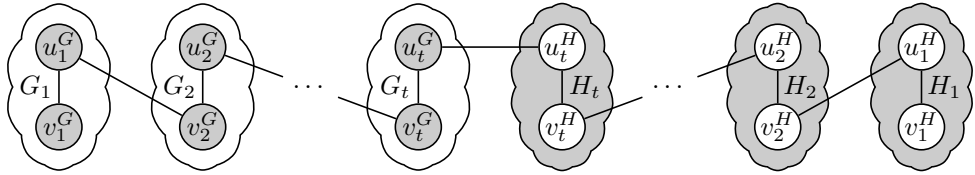

  The important feature of~$K$ is that
  every node~$w^X_i$ except for~$u^G_t$ and~$u^H_t$
  has a neighborhood equivalent to the neighborhood of
  the corresponding node~$w^X$ in the original graph~$X$.
  This is because $A$ is a non-counting automaton,
  where each node can only see the set of states of its neighbors,
  without being able to count them.
  So initially,
  the additional edges between different copies of the same graph~$X$
  do not change the “perception” of the nodes they connect.
  However,
  the two nodes~$u^G_t$ and~$u^H_t$ may have a different neighborhoods
  than~$u^G$ and~$u^H$,
  and this might affect their behavior starting at time~$1$.
  Their different behavior can be propagated to other nodes
  in subsequent rounds,
  but this propagation takes time
  before it can reach nodes in the extreme parts of the graph.

  We now construct a suitable run
  $\Run = \Tuple{\Configuration_0, \Configuration_1, \dots}$
  of~$A$ on~$K$.
  During the first~$t$ steps,
  $\Run$~tries to copy the behavior of~$\Run^G$ and~$\Run^H$.
  More precisely,
  let~$\Schedule^G$ and~$\Schedule^H$
  be schedules that schedule~$\Run^G$ and~$\Run^H$,
  respectively.
  We use them to define
  a schedule~$\Schedule$ of~$K$ that schedules~$\Run$:
  at every time $r \in \Range[0]{t}*$,
  each copied node~$w^X_i$ is selected by~$\Schedule$
  if and only if
  the original node~$w^X$ is selected by~$\Schedule^X$,
  where $i \in \Range[1]{t}$ and $X \in \Set{G, H}$.
  Note that this does not violate the strong fairness constraint
  because we have only fixed a finite prefix of~$\Schedule$.
  We can therefore extend~$\Schedule$ in such a way
  that it satisfies the strong fairness constraint.

  It remains to show that $\Run$ is neither accepting nor rejecting.
  For this,
  we prove by induction over~$r$ that for all
  $r \in \Range[0]{t}$,\, $i \in \Range[1]{t - r}$, and $X \in \Set{G, H}$,
  every copied node~$w^X_i$ in~$\Run$ at time~$r$
  is in the same state as
  the original node~$w^X$ in~$\Run^X$ at time~$r$, i.e.,
  $\Configuration_r(w^X_i) = \Configuration^X_r(w^X)$.
  This obviously holds for $r = 0$,
  since every copy~$w^X_i$ has the same label as~$w^X$.
  For $r \in \Range[1]{t}$,
  the induction hypothesis tells us that at time $r - 1$,
  each copy~$w^X_i$ with $i \in \Range[1]{t - r + 1}$
  is in the same state as~$w^X$,
  and if $i \leq t - r$,
  then $w^X_i$ also sees the same set of states as~$w^X$ in its neighborhood.
  Moreover,
  by the definition of~$\Schedule$,
  node~$w^X_i$ is selected if and only if $w^X$ is selected.
  Hence,
  provided $i \in \Range[1]{t - r}$,
  the two nodes are also in the same state at time~$r$.

  Since at time~$t$
  all nodes of~$G$ are in an accepting state in~$\Run^G$,
  and
  all nodes of~$H$ are in a rejecting state in~$\Run^H$,
  the same holds in~$\Run$ for the copies of those nodes
  in~$G_1$ (the “left-most” copy of~$G$)
  and~$H_1$ (the “right-most” copy of~$H$).
  And since $A$ is a halting automaton,
  these nodes will never change their state again.
  But this means that~$\Run$ never reaches a stable consensus,
  and therefore that it is neither accepting nor rejecting.
\end{proof}

\subsection{Proofs of Section \ref{sec:synchro}}

\ThmSynchronizer*

\begin{proof}
Let $A=(M,\SelectionConstr,\FairnessConstr)$ be a \Type{*}{*}{synchronous}{*}-automaton, and let $G=\Tuple{V, E, \lambda}$ be a graph. Let $\tilde{A}= (\tilde{M}, \tilde{\SelectionConstr},\tilde{\FairnessConstr})$, where $\tilde{M}$ is as described in Section~\ref{sec:synchro}, $\tilde{\SelectionConstr}$ is liberal, and $\tilde{\FairnessConstr}$ is weakly (strongly) fair if $\FairnessConstr$ is so. By the consistency condition, the unique run $\Run$ of $A$ on $G$ is either accepting or rejecting. By the definition of $\tilde{M}$, and since all runs of $\tilde{\FairnessConstr}(G)$ are at least weakly fair, if $\Run$ is accepting then every fair run of $\tilde{A}$ is accepting, and if $\Run$ is rejecting then every fair run of $\tilde{A}$ is rejecting. So $\tilde{A}$ also satisfies the consistency condition, and $L(A) = L(\tilde{A})$.
\end{proof}

\subsection{Proofs of Section \ref{subsec:ex-strong}}

\ThmSimulateLiberalFairnessWithExclusiveFairness*

\begin{proof}
  Given a \Type{*}{*}{liberal}{strong}-automaton~$A$,
  we construct a \Type{*}{*}{exclusive}{strong}-automaton~$B$
  such that
  for all input graphs~$G$,
  every strongly fair run of~$B$ on~$G$
  simulates a strongly fair run of~$A$ on~$G$.
  Since $A$ satisfies the consistency condition by hypothesis,
  this property implies that $B$ does too,
  and moreover that $B$ accepts a graph if and only if $A$ accepts it.
  The difficulty lies in the fact that
  $A$ and~$B$ do not share the same notion of strong fairness
  because they have different selection constraints.
  While $A$'s liberal scheduler guarantees that
  arbitrary sequences of selections will occur infinitely often,
  $B$'s exclusive scheduler can select only one node at~a~time.

  To simulate $A$'s behavior with~$B$,
  we slightly adapt the synchronizer construction
  from Section~\ref{sec:synchro}.
  Just like there,
  nodes keep track of
  their previous and current state in~$A$,
  as well as the current round number modulo~$3$.
  However, instead of updating their state in every round,
  they only do so if an additional \emph{activity flag} is set.
  Thus,
  we can simulate an arbitrary selection~$\Selection$
  by raising the flags
  of exactly those nodes that lie in~$\Selection$.
  The outcome of a round simulated in this way
  will be the same as if
  all the nodes in~$\Selection$ made a transition simultaneously.

  Now,
  the main issue is how to set the flags in each round
  in such a way that every finite sequence
  $(\Selection_1, \ldots, \Selection_n)$ of selections
  is guaranteed to occur infinitely often.
  To achieve this,
  we take advantage of the fact that $B$'s scheduler
  is strongly fair with respect to exclusive selection.
  We use the following (deterministic) rules:
  If node~$v$ is selected while it is in round~$i \bmod 3$
  and none of its neighbors are yet in round $(i + 1) \bmod 3$,
  then $v$ raises its flag;
  the next time $v$ is selected and allowed to move,
  it will simulate a transition of~$A$,
  lower its flag,
  and move to round $(i + 1) \bmod 3$.
  Otherwise,
  if $v$ is selected
  when its flag is down
  and some of its neighbors have already reached the next round,
  it simply moves to round $(i + 1) \bmod 3$
  without simulating a transition.

  Formally,
  if the machine of~$A$ is
  $M =\Tuple{Q, \delta_0, \delta, Y, N}$
  with input alphabet~$\Alphabet$ and counting bound~$\beta$,
  we define the machine of~$B$ as
  $M' =\Tuple{Q', \delta'_0, \delta', Y', N'}$,
  where
  \begin{equation*}
    Q' =
    \!{\underbrace{Q}_{\text{previous}}}\! \times
    \!{\underbrace{Q}_{\text{current}}}\! \times
    \,{\underbrace{\Set{0, 1, 2}}_{\text{round}}}\, \times
    \,{\underbrace{\Set{\bot, \top}}_{\text{flag}}},
  \end{equation*}
  $Y'$ and~$N'$ are defined analogously, and
  $\delta'_0(a) = \Tuple{\delta(a), \delta(a), 0, \bot}$
  for all $a \in \Alphabet$.
  The transition function~$\delta'$ is described as follows.
  Let $v$ be a node, and assume it is selected by the scheduler.
  \begin{itemize}
  \item In case $v$ is in state $\Tuple{q, q', i, \bot}$:
    \begin{itemize}
    \item if none of $v$'s neighbors are yet in round $(i + 1) \bmod 3$,
      then $v$ moves to $\Tuple{q, q', i, \top}$;
    \item else, if some neighbor of $v$ is still in round $(i - 1) \bmod 3$,
      then $v$ stays in $\Tuple{q, q', i, \bot}$;
    \item else, $v$ moves to state $\Tuple{q', q', (i + 1) \bmod 3, \bot}$.
    \end{itemize}
  \item In case $v$ is in state $\Tuple{q, q', i, \top}$:
    \begin{itemize}
    \item if some neighbor of $v$ is still in round $(i - 1) \bmod 3$,
      then $v$ stays in $\Tuple{q, q', i, \top}$;
    \item else, $v$ moves to  $\Tuple{q', q'', (i + 1) \bmod 3, \bot}$,
      where $q'' = \delta(q', P)$
      and $P$ is the $\beta$-bounded multiset consisting of
      the current states of the neighbors
      who are in round~$i$,
      and the previous states of the neighbors
      who are in round $(i + 1) \bmod 3$.
    \end{itemize}
  \end{itemize}

  Notice that the above construction allows the scheduler of $B$
  to choose an arbitrary selection~$\Selection$ in each round.
  For instance, the scheduler can first bring all nodes to round~$i \bmod 3$,
  next select all nodes in~$\Selection$ (one by one) to raise their flags,
  then select those same nodes again
  so that they can perform their transitions
  and move to round $(i + 1) \bmod 3$,
  and finally select all the remaining nodes
  to bring them to the next round as well.
  To prevent nodes outside of~$\Selection$ from being activated,
  the scheduler has to select them in some order
  that ensures that
  at least one of their neighbors is already in the next round
  (for example,
  a breadth-first or depth-first order
  starting from the nodes in~$\Selection$).
  Since the scheduler is strongly fair,
  by Lemma~\ref{lem:strongfairness},
  every finite sequence of selections appears infinitely often.
\end{proof}

\ThmSimulateExclusivityWithFairness*

\begin{proof}
  First,
  we note that the only way exclusivity could possibly be useful
  is to break symmetry between adjacent nodes.
  This is because for an independent set
  (i.e., a set of pairwise non-adjacent nodes),
  the order of activation is irrelevant:
  whether the scheduler activates them
  all at once or one by one in some arbitrary order,
  the outcome will always be the same.
  More precisely,
  if we consider a graph
  $G = \Tuple{V, E, \lambda}$,
  a configuration~$\Configuration$ on~$G$,
  and some \emph{independent} set of nodes $U \subseteq V$,
  then the scheduler can choose any sequence of selections
  $\Tuple{\Selection_1, \dots, \Selection_n}$
  such that
  $\bigcup_{i \in \Range[1]{n}} \Selection_i = U$
  and
  $\Card{\SetBuilder{i \in \Range[1]{n}}{v \in \Selection_i}} = 1$
  for all $v \in U$.
  Regardless of the scheduler's choice,
  the configuration~$\Configuration'$
  reached from~$\Configuration$ via the schedule fragment
  $\Tuple{\Selection_1, \dots, \Selection_n}$
  will always be the same.
  Consequently,
  to simulate a run with exclusivity,
  it suffices to simulate a run
  where no two adjacent nodes are active at the same time.

  We now describe a simple protocol
  that makes use of the strong fairness constraint
  (in an environment with liberal selection)
  to ensure that if a node wants to execute a transition,
  then it will eventually be able to do so
  while all of its neighbors remain passive.
  Suppose that an active node $v$ wants to transition
  from state $q$ to state $q'$.
  To this end,
  it first goes into an intermediate state $\Tuple{q, q'}$
  that declares this intention.
  Then,
  the next time $v$ is activated by the scheduler,
  it checks that none of its neighbors are in an intermediate state
  of the form $\Tuple{p, p'}$.
  If the check passes,
  $v$ switches to state $q'$.
  Otherwise,
  it goes back to state $q$
  and tries again the next time it is activated.
  By Lemma~\ref{lem:strongfairness},
  the strong fairness constraint guarantees
  that $v$ will infinitely often be able to execute a transition.

  More formally,
  given a \Type{*}{*}{exclusive}{strong}-automaton
  with machine
  $M =\Tuple{Q, \delta_0, \delta, Y, N}$
  and counting bound~$\beta$,
  we can simulate it by a \Type{*}{*}{liberal}{strong}-automaton
  with machine
  $M' =\Tuple{Q', \delta'_0, \delta', Y', N'}$,
  where
  \begin{equation*}
    Q' = Q \cup (Q \times Q), \qquad
    Y' = Y \cup (Y \times Y), \qquad
    N' = N \cup (N \times N),
  \end{equation*}
  $\delta'_0$ is the extension of~$\delta_0$ to the codomain~$Q'$,
  and
  $\delta'$ is defined as follows:
  For $q, q' \in Q$ and $P \in \Range{\beta}^{Q'}$
  such that $P$ contains no state $\Tuple{p, p'} \in Q'$,
  we have
  \begin{equation*}
    \delta'(q, P) = \Tuple{q, \delta(q, P)}
    \qquad \text{and} \qquad
    \delta'(\Tuple{q, q'}, P) = q',
  \end{equation*}
  and for $q, q' \in Q$ and $P \in \Range{\beta}^{Q'}$
  such that $P$ contains at least one state $\Tuple{p, p'} \in Q'$,
  we have
  \begin{equation*}
    \delta'(q, P) = q
    \qquad \text{and} \qquad
    \delta'(\Tuple{q, q'}, P) = q.
  \end{equation*}
  The first case corresponds to the situation
  where a node can make progress
  because none of its neighbors are in an intermediate state,
  whereas the second case corresponds to the situation
  where a node must wait for
  some neighbors to either complete or abort
  their current transition attempt.
\end{proof}

\subsection{Proofs of Section \ref{ssec:exclusivity-under-weak-fairness}}

\RemNonhaltingWithoutExclusivity*

\begin{proof}
  We first describe the machine of a very simple
  \Type{set}{halting}{exclusive}{weak}-automaton~$A$
  that recognizes the trivial language
  of all unlabeled graphs but relies on exclusive selection to terminate.
  It has the state set $Q = \Set{p, q, h}$,
  where $p$ is initial,
  and $h$ is halting and accepting.
  The transition function~$\delta$ is defined as follows:
  if $v$ and all its neighbors are in state~$p$,
  then $v$ moves to~$q$;
  if $v$ and all its neighbors are in state~$q$,
  then $v$ moves to~$p$;
  otherwise,
  $v$ moves to~$h$.
  For every unlabeled graph $G$,
  in the synchronous run of $A$ on $G$
  all nodes keep alternating forever between states~$p$ and~$q$
  (recall that graphs are connected and have at least two nodes),
  whereas in a run with exclusive selection,
  all nodes eventually end up in the accepting state~$h$.

  Now, using a standard product construction,
  we can easily transform any \Type{*}{*}{liberal}{weak}-automaton~$B$
  into an equivalent \Type{*}{*}{exclusive}{weak}-automaton~$C$
  whose machine never halts under synchronous execution:
  $C$ simply simulates~$A$ and~$B$ in parallel
  and accepts precisely when both accept.
\end{proof}

\LemButterflyUnfoldingConnected*

\begin{proof}
  If
  $G = \Tuple{V, E, \lambda}$
  does not contain any cycle of odd length,
  it is easy to see that
  its Kronecker cover consists of two disjoint copies of~$G$.
  Indeed,
  since containing no odd cycle is equivalent to being bipartite
  (see, e.g., \cite[Prp.~1.6.1]{Diestel17}),
  we know that $V$ can be partitioned into two sets~$V_0$ and~$V_1$
  such that every edge of~$G$ connects a node in~$V_0$ to one in~$V_1$.
  Hence,
  in the Kronecker cover~$G'$,
  we obtain one copy of~$G$ over the set of nodes
  $(V_0 \times \Set{0}) \cup (V_1 \times \Set{1})$
  and another (disjoint one) over the set
  $(V_0 \times \Set{1}) \cup (V_1 \times \Set{0})$.

  It remains to show that
  if $G$ contains an odd cycle,
  then $G'$ is connected.
  We proceed in two steps.
  First,
  consider some cycle
  $v_1 v_2 \dots v_n v_1$
  of odd length~$n$ in the original graph~$G$.
  Since $n$ is odd,
  this cycle is replicated in~$G'$ by the cycle
  \begin{equation*}
    \Tuple{v_1, 0} \, \Tuple{v_2, 1} \dots \Tuple{v_n, 0} \,
    \Tuple{v_1, 1} \, \Tuple{v_2, 0} \dots \Tuple{v_n, 1} \, \Tuple{v_1, 0}
  \end{equation*}
  of length~$2n$.
  (If $n$ were even,
  we would get two disjoint cycles of length~$n$ instead.)
  Second,
  since $G$ is connected,
  for any node $u \in V$
  there exists a path $w_1 w_2 \dots w_m$ in~$G$
  such that $w_1 = u$ and $w_m = v_1$.
  This path is replicated in~$G'$ by the two paths
  \begin{equation*}
    \Tuple{w_1, 0} \, \Tuple{w_2, 1} \dots \Tuple{w_m, i}
    \qquad\text{and}\qquad
    \Tuple{w_1, 1} \, \Tuple{w_2, 0} \dots \Tuple{w_m, j},
  \end{equation*}
  where $\Tuple{i, j} = \Tuple{1, 0}$ if $m$ is even,
  and $\Tuple{i, j} = \Tuple{0, 1}$ if $m$ is odd.
  This means that
  both $\Tuple{u, 0}$ and $\Tuple{u, 1}$ are connected
  to the aforementioned cycle of length~$2n$,
  and since $u$ was chosen arbitrarily,
  it follows that $G'$ is connected.
\end{proof}

\LemButterflyUnfoldingInvariance*

\begin{proof}
  It suffices to prove the claim for
  \Type{*}{*}{synchronous}{weak}- and \Type{*}{*}{exclusive}{weak}-automata,
  since \Type{*}{*}{liberal}{weak}-automata
  can be regarded as a special case of both.
  In the following,
  let
  $A = \Tuple{M, \Scheduler}$
  and
  $G = \Tuple{V, E, \lambda}$.
  Since~$G$ is non-bipartite
  (i.e., it contains a cycle of odd length),
  we know by Lemma~\ref{lem:butterfly-unfolding-connected}
  that its Kronecker cover
  $G' = \Tuple{V', E', \lambda'}$
  is connected
  and therefore qualifies as valid input for~$A$.

  Let us begin with the case where~$A$ is synchronous,
  i.e., a \Type{*}{*}{synchronous}{weak}-automaton,
  and consider the (unique) runs~$\Run$ and~$\Run'$ of~$A$ on~$G$ and~$G'$,
  respectively.
  Recall that
  $V' = V \times \Set{0, 1}$.
  For every node~$v$ of~$G$,
  its copies $\Tuple{v, 0}$ and $\Tuple{v, 1}$ in~$G'$
  have the same label as~$v$
  and an equivalent multiset of neighbors
  (i.e., all their neighbors are copies of $v$'s neighbors).
  It is thus easy to see by induction
  that in every round $i \in \Naturals$,
  $\Tuple{v, 0}$ and $\Tuple{v, 1}$ are in the same state in~$\Run'$
  as $v$ is in~$\Run$.
  Therefore,
  the $i$-th configuration of~$\Run'$ is accepting
  if and only if
  the $i$-th configuration of~$\Run$ is accepting,
  and hence $A$~accepts~$G'$ precisely if it accepts~$G$.

  We now turn to the case where $A$ is a \Type{*}{*}{exclusive}{weak}-automaton.
  Consider any schedule
  $\Schedule = (\Selection_0, \Selection_1, \ldots) \in (2^V)^\omega$
  that satisfies the constraints of the scheduler~$\Scheduler$.
  To prove the claim,
  it suffices to show that there exists a schedule~$\Schedule'$ of~$G'$
  that also satisfies the constraints of~$\Scheduler$
  such that
  the run~$\Run'$ of~$A$ on~$G'$ scheduled by~$\Schedule'$ is accepting
  if and only if
  the run~$\Run$ of~$A$ on~$G$ scheduled by~$\Schedule$ is accepting.
  Indeed,
  by the consistency condition,
  this implies that
  $A$ accepts~$G$ if and only if it accepts~$G'$.

  We choose
  $\Schedule' =
  (\Selection'_0, \Selection'_1, \ldots) \in (2^{V'})^\omega$
  such that
  \begin{equation*}
    \Selection'_{2t} = \Selection_i \times \Set{0}
    \quad \text{and} \quad
    \Selection'_{2t+1} = \Selection_i \times \Set{1}
  \end{equation*}
  for all $t \in \Naturals$.
  That is,
  for every node~$v$ of~$G$,
  if~$v$ is active at time~$t$,
  then its copy $\Tuple{v, 0}$ in~$G'$ is active at time~$2t$,
  and its copy $\Tuple{v, 1}$ is active at time~$2t+1$.
  Note that since~$\Schedule$ is weakly fair,
  so is~$\Schedule'$.
  Furthermore,
  the exclusivity of~$\Schedule$
  also carries over to~$\Schedule'$
  (this is why we do not schedule
  $\Tuple{v, 0}$ and $\Tuple{v, 1}$ simultaneously).
  However,
  $\Schedule'$ is not strongly fair in general,
  and therefore the assumption
  that $A$ is a \Type{*}{*}{*}{weak}-automaton
  is essential.

  Now,
  since $\Tuple{v, 0}$ and $\Tuple{v, 1}$ are not connected,
  and since both have the same label as~$v$
  and an equivalent multiset of neighbors,
  it is again easy to see by induction that the following holds:
  at every even time~$2i$,
  both copies are in the same state as~$v$ is at time~$i$,
  while at every odd time $2i+1$,
  copy $\Tuple{v, 0}$ is already in the same state as~$v$ at time~$i+1$,
  but copy $\Tuple{v, 1}$ is still in the state $v$ had at time~$i$.
  Here we rely on the fact that
  each selection~$\Selection_i$ is a singleton,
  which ensures that
  if $v$ is active in~$G$ at time~$i$,
  then no other node is active at the same time.
  This means that $\Tuple{v, 0}$ and $\Tuple{v, 1}$
  receive the same multiset of states from their neighbors in~$G'$
  at times $2i$ and $2i+1$, respectively.
  Consequently,
  the $(2i)$-th configuration of~$\Run'$ is accepting
  if and only if
  the $i$-th configuration of~$\Run$ is accepting,
  and the $(2i + 1)$-th configuration of~$\Run'$ is accepting
  if and only if
  both the $i$-th and the $(i + 1)$-th configurations of~$\Run$ are accepting.
  Given that legal runs must eventually reach a stable consensus
  (i.e., only accepting or only rejecting configurations after a certain time),
  this means that
  $\Run'$~is accepting if and only if $\Run$~is accepting.
\end{proof}


\ThmSimulateExclusiveWeakWithLiberalWeak*

\begin{proof}
  In the following,
  we show how,
  for a given \Type{*}{*}{exclusive}{weak}-automaton~$A$,
  we can construct an equivalent \Type{*}{*}{synchronous}{weak}-automaton~$B$
  (i.e., a synchronous automaton).
  This is sufficient to prove the claim
  because we know from Theorem~\ref{lem:synchronizer} that
  $B$ can always be simulated
  by a \Type{*}{*}{liberal}{weak}-automaton using a synchronizer.

  First of all,
  let us observe that the task would be straightforward
  if we were guaranteed that
  the labels of the input graph define a proper vertex coloring.
  Indeed,
  since each color of a proper coloring
  represents an independent set,
  $B$~could simply operate in cyclically repeating phases
  that correspond to the different colors.
  More precisely,
  if the given colors were $0, \dots, {k - 1}$,
  then in the $i$-th round
  (i.e., the $i$-th time all nodes change state synchronously),
  only the ${(i \mod k)}$-colored nodes
  would evaluate the transition function
  of the simulated automaton~$A$.
  As explained in the first paragraph of the proof of
  Theorem~\ref{thm:simulate-exclusivity-with-fairness},
  such a run is equivalent to a run of an exclusive scheduler
  that activates the nodes in each independent set one by one
  (in some arbitrary order).

  Obviously the above approach only works
  if we are given a proper coloring.
  Nevertheless,
  it can be adapted to a special case of uncolored graphs:
  if the input graph happens to be bipartite,
  then there exist exactly two possible 2-colorings.
  This is because
  as soon as we fix the color of a single node,
  there is only one possible choice of color
  for all the remaining nodes.
  However,
  choosing one of the two 2-colorings would require to break symmetry,
  which a \Type{*}{*}{synchronous}{weak}-automaton cannot do.
  So instead,
  we simply work with both colorings in parallel.

  We now go into more details on how to simulate
  a \Type{*}{*}{exclusive}{weak}-automaton~$A$
  by a \Type{*}{*}{synchronous}{weak}-automaton~$B$
  \emph{on bipartite graphs}.
  Let
  $M =\Tuple{Q, \delta_0, \delta, Y, N}$
  be the machine of~$A$
  with input alphabet~$\Alphabet$ and counting bound~$\beta$,
  and let $\Set{0, 1}$ be a set of colors
  that we will use to color the graph.
  At any point in time in an execution of~$B$,
  each node $v$ stores a pair of states
  $\Tuple{q_0, q_1} \in Q \times Q$,
  where $q_0$ represents $v$'s current state
  in case its color is~$0$,
  and similarly $q_1$ represents $v$'s current state
  in case its color is~$1$.
  This way,
  $B$ can run the aforementioned round-based simulation of~$A$
  for both possible 2-colorings in parallel.
  To simulate the case where $v$ is $0$-colored,
  $v$~looks at the state in its own $0$-component
  but at the states in its neighbors' $1$-component
  (since the neighbors must be $1$-colored if $v$ is $0$-colored).
  To simulate the case where $v$ is $1$-colored,
  the procedure is the other way around.

  More formally,
  the machine of~$B$ can be defined as
  $M' =\Tuple{Q', \delta'_0, \delta', Y', N'}$,
  where
  \begin{equation*}
    Q' = Q \times Q \times \Set{0, 1}, \qquad
    Y' = Y \times Y \times \Set{0, 1}, \qquad
    N' = N \times N \times \Set{0, 1},
  \end{equation*}
  $\delta'_0(a) = \bigTuple{\delta(a), \delta(a), 0}$
  for all $a \in \Alphabet$,
  and the transition function $\delta'$ is defined as follows,
  for $q_0, q_1 \in Q$ and $P \in \Range{\beta}^{Q'}$:
  \begin{align*}
    \delta' \bigTuple{\Tuple{q_0, q_1, 0}, P} &=
    \bigTuple{
      \delta(q_0, P_1),\,
      q_1,\,
      1
    }, \\
    \delta' \bigTuple{\Tuple{q_0, q_1, 1}, P} &=
    \bigTuple{
      q_0,\,
      \delta(q_1, P_0),\,
      0
    },
  \end{align*}
  where~$P_0$ and~$P_1$ are the $\beta$-bounded projections of~$P$
  to the two first state components, i.e.,
  \begin{align*}
    P_0 \colon p
    & \mapsto \textstyle
      \min \bigSet{\beta,\, \sum_{p_1 \in Q,\, i \in \Set{0,1}} P(p, p_1, i)}, \\
    P_1 \colon p
    & \mapsto \textstyle
      \min \bigSet{\beta,\, \sum_{p_0 \in Q,\, i \in \Set{0,1}} P(p_0, p, i)},
  \end{align*}
  for all $p \in Q$.
  The third state component
  counts the number of synchronous rounds modulo~$2$.
  If the round number is even,
  each node behaves as if it were $0$-colored
  and its neighbors were $1$-colored.
  Thus,
  each node updates its $0$-component
  according to its neighbors' $1$-components.
  Meanwhile,
  the $1$-component remains unchanged
  because $1$-colored nodes are
  supposed to remain passive in even rounds.
  If the round number is odd,
  everything is the other way around.

  The above construction of~$B$ is based on the assumption
  that the input graph is bipartite.
  However,
  we now argue that in fact this assumption is not necessary.
  To do so,
  we have to distinguish two cases:
  \begin{itemize}
  \item If the input graph~$G$ is bipartite,
    then by construction,
    the synchronous run of~$B$ on~$G$
    simulates in parallel two runs of~$A$ on~$G$ with exclusive selection.
    By the consistency condition,
    this implies that
    $G$ is accepted by~$B$ if and only if it is accepted by~$A$.
  \item If~$G$ is not bipartite,
    then by Lemma~\ref{lem:butterfly-unfolding-connected},
    its Kronecker cover~$G'$ is connected
    and therefore constitutes a legal input for a distributed automaton.
    Now,
    by Lemma~\ref{lem:butterfly-unfolding-invariance},
    $B$~accepts~$G$ if and only if it accepts~$G'$.
    Since $G'$ is bipartite
    (by the definition of a Kronecker cover),
    we know from the above discussion that
    $B$~accepts~$G'$ if and only if $A$~accepts~$G'$.
    Finally,
    again by Lemma~\ref{lem:butterfly-unfolding-invariance},
    $A$~accepts~$G'$ if and only if it accepts~$G$.
    From this chain of equivalences,
    we can conclude that
    $G$ is accepted by~$B$ if and only if it is accepted by~$A$.
  \end{itemize}
  Notice that in the case where the input graph is not bipartite,
  $B$~simulates~$A$ on the Kronecker cover~$G'$
  instead of the actual graph~$G$.
  So in some sense,
  our construction only performs a “pseudo simulation”,
  where the simulated run
  may not correspond to any possible run on~$G$.
  Nevertheless,
  this is sufficient because
  \Type{*}{*}{*}{weak}-automata cannot distinguish between~$G$ and~$G'$.
\end{proof}

\subsection{Proofs of Section \ref{sec:separations}}

\PropAsFRecogStars*
\begin{proof}
\newcommand{\NE}{\textit{NE}}
We first present a \Type{set}{stabilizing}{liberal}{strong}-automaton that recognizes $S$. 
The states of the automaton are pairs $(d, c)$, where $d \in \{\textit{leaf}, \textit{center}, \textit{unknown}, \textit{neither}\}$ is the \emph{estimate} of $v$, and $c \in \{ 0,1 \}$ is its \emph{color}. The accepting states are those with estimate \textit{leaf} or \textit{center}, and the rejecting states are those with estimate \textit{unknown} or \textit{neither}. Initially all nodes are in state $(\textit{unknown}, 0)$. Let $(d, c)$ be the current state of a node $v$, and let $\NE(v)$ denote the current set of estimates of the neighbors of $v$.
If $v$ is selected by the scheduler, then it moves to the state $(d', c')$, where $c'=1-c$,  and $d'$ is given by:

\begin{description}
\item[(a)] If $\textit{neither} \in \NE(v)$, then  $d'=\textit{neither}$.
\item[(b)] If $\textit{neither} \notin \NE(v)$, $d  = \textit{unknown}$, $\textit{center} \notin \NE(v)$,  and at least two neighbors of $v$ have different colors, then $d'=\textit{center}$.
\item[(c)] If $\textit{neither} \notin \NE(v)$, $d  = \textit{unknown}$, $\textit{center} \in \NE(v)$,  and at least two neighbors of $v$ have different colors, then $d'=\textit{neither}$.
\item[(d)] If $\textit{neither} \notin \NE(v)$, $d  = \textit{unknown}$, $\NE(v) = \Set{\textit{center}}$,  and all neighbors of $v$ have the same color, then $d' = \textit{leaf}$.
\item[(e)] If $\textit{neither} \notin \NE(v)$, $d  = \textit{center}$, and $\textit{center} \in \NE(v)$, then  $d'=\textit{neither}$.
\item[(f)] If $\textit{neither} \notin \NE(v)$, $d = \textit{leaf}$, and at least two neighbors of $v$ have different colors, then $d'=\textit{neither}$.
\item[(g)] Otherwise $d' = d$.
\end{description}

Assume that $G$ is not a star. If it consists of exactly two nodes connected by an edge, then it is easy to see that the estimate of both nodes remains forever \textit{unknown}, so $G$ is rejected. Otherwise, $G$ contains at least one edge $\{u, v\}$ such that both $u$ and $v$ have degree at least 2. We show that eventually at least one of~$u$ and~$v$ reaches estimate \textit{neither}. By (a), every node eventually reaches estimate \textit{neither}, and so $G$ is rejected.

First we claim that both $u$ and $v$ eventually reach states with estimate  \textit{center} or \textit{neither}. This is the point at which we make crucial use of strong fairness: by Lemma~\ref{lem:strongfairness}, it ensures that $v$ is eventually selected in a configuration \emph{in which at  least two neighbors of $v$ have different colors}. If in this configuration $v$ has estimate \textit{unknown}, then $v$ moves either to \textit{neither} (cases~(a) and ~(c)) or \textit{center} (case (b)), and  if it has estimate \textit{leaf}, then $v$ moves to \textit{neither} (cases~(a) and ~(f)). The same holds for $u$, and so the claim is proved. 

By the claim, at least one of $u$ and $v$ eventually reaches estimate \textit{neither}, in which case we are done, or both eventually reach \textit{center}; in this case, the next time one of the two is selected it moves to \textit{neither} (case (e)), and we are also done.

Assume now that $G$ is a star. We show that every node ends up with estimate \textit{leaf} or \textit{center}. Since leaves have only one neighbor, cases (b), (c), and (f) never apply, and so they can never reach estimate \textit{center}. This implies that case (e) also never applies for leaves. Further, as long as the center has estimate \textit{unknown}, all leaves remain in \textit{unknown}, because (a) and (d) do not apply. It follows that the center also remains in \textit{unknown} until it is selected in a configuration in which at least two neighbors have different colors, which eventually happens by strong fairness; at that moment it moves to \textit{center} (case (b)). Since (e) never applies, the center maintains the estimate \textit{center} forever. Once the center has reached estimate \textit{center}, whenever a leaf is selected it changes its estimate to \textit{leaf} (case (d)). After that, no other rule than (g) ever applies, and so the leaf maintains estimate \textit{leaf} forever. This concludes the proof of the first part of the proposition.

\medskip

For the second part we present a \Type{multiset}{halting}{liberal}{weak}-automaton with counting bound $\beta=2$ that recognizes $S$. We only sketch the automaton, since the ability to count makes the task of recognizing $S$ easy. 
Recall that $\beta=2$ means that for each state $q$ a node can detect if it has zero, exactly one, or at least two neighbors in $q$. 

The states of the automaton are $\{ \textit{init}, \textit{leaf}, \textit{non-leaf}, \textit{accept}, \textit{reject} \}$. The yes and no states are \textit{accept} and \textit{reject},  respectively

Initially all nodes are in state $\textit{init}$. Let $v$ be a node. Observe that, since the automaton can count, a selected node can directly observe if it is a leaf or not. When $v$ is selected:
\begin{description}
\item[(a)] If $v$ has only one neighbor, then
\begin{description}
\item[(a.1)] if the neighbor is in state \textit{init} or \textit{non-leaf}, $v$ moves to \textit{leaf};
\item[(a.2)] if the neighbor is in state \textit{leaf} or \textit{reject}, $v$ moves to \textit{reject}; and
\item[(a.3)] if the neighbor is in state \textit{accept}, $v$ moves to state \textit{accept}.
\end{description}
\item[(b)] If $v$ has more than one neighbor, then
\begin{description}
\item[(b.1)] if at least one neighbor is in state \textit{reject} or \textit{non-leaf}, $v$ moves to \textit{reject};
\item[(b.2)] else if at least one neighbor is in state \textit{init}, $v$ moves to \textit{non-leaf};
\item[(b.3)] else (all neighbors in states \textit{leaf} or \textit{accept}), $v$ moves to \textit{accept}.
\end{description}
\end{description}

Assume $G$ is a star. By (a.1) and (b.2), a node can only reach state \textit{leaf} (\textit{non-leaf}) if it really is a leaf (non-leaf) of $G$. This fact, together with an inspection of (a.2) and (b.1), shows that a node can only reach \textit{reject} if $G$ is not a star. Further inspection of (a.3) and (b.3) shows that it can only reach state \textit{accept} if $G$ is a star. So it only remains to prove that every node eventually reaches \textit{accept} or  \textit{reject}. By (a.2) and (a.3) it suffices to show that eventually some node reaches \textit{accept} or  \textit{reject}. If all nodes are leaves, then there are at most two nodes, and by (a.2) they eventually move to \textit{reject}. Assume now that there is at least one non-leaf. By (a), (b), and weak fairness, eventually all nodes leave state \textit{init}, and so  all non-leaves are in one of \textit{non-leaf}, \textit{accept}, or \textit{reject}. If at least one non-leaf is in \textit{accept} or \textit{reject}, we are done. Otherwise, if $G$ is a star, then by (b.3) the (unique) non-leaf eventually moves to \textit{accept}; if $G$ is not a star, then two neighbors are in state \textit{non-leaf}, and by (b.1) the next time any of them is selected it moves to \textit{reject}.
\end{proof}

\PropCthreeSeparation*

\begin{proof}
\noindent (a) $\Triangle$ is recognizable by a \Type{set}{stabilizing}{liberal}{strong}-automaton. \\
We sketch the behavior of a \Type{set}{stabilizing}{liberal}{strong}-automaton for $\Triangle$. Recall that the nodes of the cycle $\triangle$ are labeled by $0$, $1$, and $2$. First, if a node with label $i$ detects that it has more than two neighbors, or that the set of labels of its neighbors is different from $\Set{{(i-1) \bmod 3}, {(i+1) \bmod 3}}$, then the node moves to a rejecting state. Nodes with a neighbor in a rejecting state also move to a rejecting state. To detect that a node has more than two neighbors, the automaton uses the same trick as in Proposition \ref{prop:AsF-recog-stars}: the state of each node has a \emph{color} component with three possible values, which changes whenever the node is active. By strong fairness and Lemma~\ref{lem:strongfairness}, if the node has more than two neighbors, then it will eventually see that its neighbors have three different colors, and reject. 

As we consider only connected graphs,
the preceding tests ensure that graphs which are not cycles with cyclic labeling ${0{-}1{-}2}$ are eventually rejected.
It remains to ensure that a cycle of length other than 3 is eventually rejected too. For this, the automaton checks an equivalent condition:  the cycle contains exactly one node labeled by $2$.  Nodes labeled by $2$ alternate between two phases, 0 and 1. In phase $b \in \{0,1\}$, the node asks its neighbor labeled by $b$ to propagate a signal through the cycle, and then waits until a signal arrives. (For this, the node moves to a state indicating that it wants the signal to be propagated, and waits for the neighbor to reach a state indicating it has received the message.) If the next signal arrives through the $(1-b)$ neighbor, the node moves to phase $(1-b)$; if it arrives through the $b$ neighbor, the node moves to a rejecting state. If the cycle contains only one node labeled by~$2$, then every signal sent through one neighbor arrives through the other. However, if the cycle contains at least two nodes labeled by~$2$, then by strong fairness, eventually two consecutive $2$-nodes send a clockwise and a counterclockwise signal, and so eventually a $2$-node sends a signal through a node, receives the next signal through the same node, and moves to the rejecting state.

\medskip

\noindent (b) $\Triangle$ is not recognizable by \Type{multiset}{stabilizing}{*}{weak}-automata. \\
Let $\hexagon$ be the hexagon whose nodes are labeled by ${0{-}1{-}2{-}0{-}1{-}2}$ (and back to $0$). We show that every \Type{multiset}{stabilizing}{*}{weak}-automaton $A$ that accepts~$\triangle$ also accepts~$\hexagon$. For this, consider the synchronous schedules~$\Schedule_3$ and~$\Schedule_6$ of~$A$ on~$\triangle$ and~$\hexagon$. Observe that  $\Schedule_3$ and $\Schedule_6$ are weakly fair, and so the runs $\Run_3=(\Configuration_{3,0}, \Configuration_{3,1} \cdots)$ and $\Run_6=(\Configuration_{6,0}, \Configuration_{6,1} \cdots)$ scheduled by them are fair too.  By the consistency condition, $\Run_3$ is accepting. Let $v, v'$ be nodes of $\triangle$ and $\hexagon$, respectively, carrying the same label. 
It is easy to see that $\Configuration_{3,t}(v) = \Configuration_{6,t}(v')$ for every time $t \geq 0$.
So $\Run_6$ is also accepting, and thus, by the consistency condition, $A$ accepts~$\hexagon$.

\medskip

\noindent (c) $\Triangle$ is not recognizable by 
\Type{multiset}{halting}{*}{strong}-automata. \\
We proceed as in part (b): we show that every \Type{multiset}{halting}{*}{strong}-automaton $A$ that accepts~$\triangle$ also accepts~$\hexagon$. Let $\Schedule_3 = (\Selection_{3,0}, \Selection_{3,1}, \ldots)$ be a strongly fair schedule of~$A$ on~$\triangle$, and let $\Run_3=(\Configuration_{3,0}, \Configuration_{3,1} \cdots)$ be the run scheduled by it. Since $\triangle$ is accepted, $\Run_3$ is accepting, and so there is a configuration 
$\Configuration_{3,t_0}$ in which every agent is in an accepting state. 

For every $1 \leq t \leq t_0$, let $\Selection_{6,t}$ be the selection that for every label $\ell = 0,1,2$
contains the two nodes of $\hexagon$ labeled by $\ell$ if{}f $\Selection_{3,t}$ contains the node of $\triangle$ labeled by $\ell$ (loosely speaking, $\Selection_{6,t}$ ``duplicates'' $\Selection_{3,t}$). Let $\Schedule_6$ be 
the result of choosing an arbitrary strongly fair schedule $(\Selection_{6,0}', \Selection_{6,1}', \ldots)$
of~$A$ on~$\hexagon$, and replacing $\Selection_{6,0}', \ldots, \Selection_{6,t_0}'$ by 
$\Selection_{6,0}, \ldots, \Selection_{6,t_0}$. Since $\Schedule_6$ satisfies the definition of strong fairness, the run $\Run_6=(\Configuration_{6,0}, \Configuration_{6,1} \cdots)$ scheduled by it is also strongly~fair.

Let $v, v'$ be nodes of $\triangle$ and $\hexagon$, respectively, carrying the same label.  By the definition of the selection  $\Selection_{6,t}$ for $1 \leq t \leq t_0$, we have $\Configuration_{3,t_{0}}(v) = \Configuration_{6,t_0}(v')$.  So, in particular,  every node of $\Configuration_{6,t_0}(v')$ is in an accepting state. Since $A$ is a halting automaton, nodes that have accepted can no longer change their state, so $\Run_6$ is accepting, and therefore $A$ accepts $\hexagon$.
\end{proof}

\PropCasFRecogStars*

\begin{proof}
In Proposition \ref{prop:AsF-recog-stars} we have exhibited a \Type{multiset}{halting}{liberal}{weak}-automaton $A$ recognizing $\Stars$. We now give a  \Type{multiset}{halting}{exclusive}{strong}-automaton $B$ with $\beta=2$ that uses counting, exclusivity, and strong fairness to further decide if the number of leaves is even. Loosely speaking, $B$ first executes $A$; if $A$ rejects, then $B$ rejects, because the graph is not even  a star. If $A$ accepts, then $B$ enters a new phase during which it counts the number of leaves modulo 2. 
By Theorem~\ref{thm:simulate-exclusivity-with-fairness},
$B$ is equivalent to a \Type{multiset}{halting}{liberal}{strong}-automaton.

We can assume that when  $A$ accepts, all nodes are labeled with either \textit{leaf} or
\textit{center} (the unique non-leaf). We first give an informal description of $B$. Leaves can be in states \textit{visible}, \textit{invisible}, \textit{dead}, \textit{even}, or \textit{odd}. Intuitively, while leaves have not been counted by the center, they alternate between the states \textit{visible} and \textit{invisible}. The center only increments its modulo-$2$ counter if exactly one leaf is \textit{visible}. After a leaf is counted, it moves to \textit{dead}. When all leaves become dead, i.e., when they have all been counted, the center decides whether to accept or reject; the leaves read the decision from the counter, and move to  \textit{even} or \textit{odd} accordingly. 

Formally, the state of a leaf is one out of  $\{ \textit{visible}, \textit{invisible}, \textit{dead}, \textit{even}, 
 \textit{odd} \}$, where \textit{even} is accepting, and \textit{odd} is rejecting. Initially all leaves are invisible. The states of the center are of the form 
 $$(ph,p,d) \in \{0,1,2\} \times \{0,1\} \times \{\textit{none}, 0, 1\},$$ 
 \noindent where $ph$ is the phase, $p$ the parity, and $d$ the decision, respectively. The initial state is $(0, 0, \textit{none})$, and the accepting and rejecting states are those with decision~$0$ and~$1$, respectively. The transition function is as follows. Let $v$ be a node selected by the scheduler.
\begin{itemize}
\item If $v$ is a leaf, and its current state is~$s$, then:
\begin{itemize}
\item If $s = \textit{invisible}$ (\textit{visible}) and the center is in phase $0$, then $v$ moves to \textit{visible}  (\textit{invisible}). \\
Intuitively, while the center is in phase $0$, $v$ keeps making itself visible and invisible to the center.
By Lemma~\ref{lem:strongfairness}, strong fairness guarantees that eventually exactly one leaf will be visible to the center.
\item If $s = \textit{visible}$ and the center is in phase $1$, then $v$ moves to \textit{dead}. \\
Intuitively, $v$ knows that it has been counted by the center, and dies.
\item If $s = \textit{dead}$ and the center is in phase $2$, then $v$ moves to \textit{even} or \textit{odd}, depending on the decision made by the center. 
\item Otherwise $v$ remains in state $s$.
\end{itemize}
\item If $v$ is the center, and its current state is $\alpha = (ph,p,d)$, then $v$ changes its state as follows:
\begin{itemize}
\item If exactly one leaf is visible and $ph=0$, then the center moves to $\alpha[ph \rightarrow 1, p \rightarrow 1-p]$. \\
(Where $\alpha[ph \rightarrow 1, p \rightarrow 1-p]$ denotes the result of substituting $1$ for $ph$ and $1-p$ for $p$ in $\alpha$.)
Intuitively, the center counts the visible leaf.
Since the scheduler is exclusive,
no other leaf can change its visibility status
at the same time as the center performs this operation.
This guarantees that multiple leaves are not counted as one,
and that the unique counted leaf remains visible.
\item If all leaves are invisible or dead, at least one leaf is invisible, and $ph=1$, then the center 
moves to $\alpha[ph \rightarrow 0]$. \\
Intuitively, after counting a leaf the center sees that the leaf knows it has been counted and died.
\item If all leaves are dead and $ph=1$, then the center 
moves to $\alpha[ph \rightarrow 2, d \rightarrow p]$. \\
Intuitively, the counting is done, and the center takes the current parity as the decision.
\item Otherwise the center remains in state $\alpha$.
\end{itemize}
\end{itemize}

In every strongly fair run, eventually the center is selected in a configuration in which exactly one 
leaf, say $v$ is visible. This is detected by the center, which updates its counter and moves to phase $1$. 
The center stays in phase 1 until it sees that all leaves are invisible or dead, which guarantees that $v$ knows it has been counted and died. The center then moves to phase $0$ again, to count the next leaf. When all leaves have been counted (which the center can detect by observing that they are all dead), the center knows that its parity bit is the correct one, and moves to phase 2. By fairness, all leaves eventually read the result from the center, and move to \textit{even} or \textit{odd}.

Notice how the use of an exclusive scheduler simplifies our design.
Indeed,
the distributed machine described above
would not be correct under a liberal scheduler,
because the center could be deceived as follows.
Let~$u$ be the center,
and let~$v_1$ and~$v_2$ be two leaves.
Suppose that $u$ is in phase~$0$ and $v_1$ is the only visible leaf.
Next,
$u$ and~$v_2$ are selected \emph{simultaneously},
so $u$ moves to phase~$1$ and increments its counter by~$1$
(as it sees exactly one visible leaf),
while~$v_2$ becomes visible
(as it sees the center in phase~$0$).
Now both~$v_1$ and~$v_2$ will die
(as they are visible and~$u$ is in phase~$1$),
but only~$v_1$ has been counted.
In order to avoid such problems,
we could introduce an additional verification phase in which
the center checks that it has counted exactly one leaf,
but this would make the protocol more complicated.
So instead,
we first take exclusivity for granted,
and then implement it using the construction of
Theorem~\ref{thm:simulate-exclusivity-with-fairness}.
\end{proof}

\PropCASFCannotEvenStars*

\begin{proof}
  For the sake of obtaining a contradiction,
  let us assume that there exists
  a \Type{set}{stabilizing}{liberal}{strong}-automaton $A$
  with machine
  $M =\Tuple{Q, \delta_0, \delta, Y, N}$
  that recognizes $\EvStars$.
  We must first introduce several concepts related to~$M$
  before we can get to the actual contradiction argument.

  Without loss of generality,
  we assume that the language of star graphs is
  $\Stars = \{ \StarGraph_i \mid i \geq 2 \}$,
  where $\StarGraph_i$ is the unlabeled graph with nodes
  $\{ r, l_1, \ldots, l_i \}$
  and edges
  $\{ r, l_1\}, \ldots, \{r, l_i\}$.
  We call $r$ the \intro{root} and
  $l_1, \ldots, l_i$ the \intro{leaves} of the star.
  Throughout this proof,
  we consider only  configurations of~$M$
  whose underlying graph is $\StarGraph_i$ for some $i \geq 2$,
  and call them \intro{star configurations}.
  For notational simplicity,
  we sometimes identify a star configuration with a tuple
  $\Configuration = \Tuple{q, f}$,
  where
  $q \in Q$
  is the state of $r$
  and
  ${f \colon Q \to \Naturals}$
  is a function that assigns to each state~$p$
  the number of leaves of~$\StarGraph_i$ that are in state~$p$.
  We denote the total number of nodes of~$\Configuration$
  by~$\Card{\Configuration}$,
  i.e., $\Card{\Configuration} = 1 + \sum_{p \in Q} f(p)$.
  Clearly, a configuration
  $\Configuration$ of $\StarGraph_i$
  satisfies $\Card{\Configuration} = i+1$.

  A \intro{base configuration} is a star configuration
  in which every state $p \in Q$
  occurs at most once on a leaf node.
  We write $\mathit{Base}$ for the set of all base configurations,
  i.e., $\mathit{Base} = Q \times \Set{0, 1}^Q$.
  The base configuration associated with
  $\Configuration = \Tuple{q, f}$ is the configuration
  $\operatorname{base}(\Configuration) = \Tuple{q, f'}$
  such that
  $f'(p) = \min \Set{f(p), 1}$
  for all $p \in Q$.
  Intuitively,
  $\operatorname{base}(\Configuration)$
  is the smallest star configuration
  in which the root sees the same set of states as in~$\Configuration$.

  Given two configurations $\Configuration = \Tuple{q, f}$ and
  $\Configuration' = \Tuple{q', f'}$,
  we let $C \preceq C'$ denote that
  $q = q'$, $f(p) \leq f'(p)$ for all $p \in Q$, and ${f(p) = 0}$ if and only if
  ${f'(p) = 0}$.
  Observe that $\preceq$ is a partial order.
  The \intro{upward closure} of~$\Configuration$
  is the set $\UpwardClosure{\Configuration} \DefEq \{ C' \mid C' \succeq C\}$.
  In other words,
  $\UpwardClosure{\Configuration}$ is the set of configurations
  that one can obtain by duplicating some leaves of~$\Configuration$.
  Notice that the root of such a configuration also sees
  the same set of states as in~$\Configuration$.

  The successor relation on configurations of~$M$
  will be denoted by~$\rightarrow$.
  That is,
  for two configurations~$\Configuration[1]$ and~$\Configuration[2]$,
  we write
  $\Configuration[1] \rightarrow \Configuration[2]$
  if and only if
  $\Configuration[1]$ can reach~$\Configuration[2]$
  in a single execution step of~$M$.
  (This means that there exists
  a selection~$\Selection$ of $\Configuration[1]$'s underlying graph
  such that one obtains~$\Configuration[2]$
  by evaluating $M$'s transition function~$\delta$
  at the nodes of~$\Configuration[1]$ selected by~$\Selection$.)
  We lift this relation to sets of configurations~%
  $\ConfigurationSet[1]$ and~$\ConfigurationSet[2]$
  in a rather natural way,
  writing
  $\ConfigurationSet[1] \rightarrow \ConfigurationSet[2]$
  if and only if
  for every $\Configuration[1] \in \ConfigurationSet[1]$
  there exists some $\Configuration[2] \in \ConfigurationSet[2]$
  such that
  $\Configuration[1] \rightarrow \Configuration[2]$.
  Furthermore,
  we use the standard notation~$\rightarrow^*$
  for the reflexive-transitive closure of~$\rightarrow$,
  and~$\rightarrow^i$
  for the $i$-fold composition of~$\rightarrow$ with itself,
  where $i \in \Naturals$.

  \smallskip
  \noindent\textbf{Claim~1.}\,
  If
  $\Configuration[1] \rightarrow^* \Configuration[2]$,
  then
  $\UpwardClosure{\Configuration[1]} \rightarrow^* \UpwardClosure{\Configuration[2]}$. \\
  Proceeding by induction over $i \in \Naturals$,
  we show that
  $\Configuration[1] \rightarrow^i \Configuration[2]$
  implies
  ${\UpwardClosure{\Configuration[1]} \rightarrow^i \UpwardClosure{\Configuration[2]}}$.
  The case $i = 0$ is trivial,
  since
  $\Configuration[1] \rightarrow^0 \Configuration[2]$
  means that
  $\Configuration[1] = \Configuration[2]$.

  For $i = 1$,
  we observe that for every configuration
  $\Configuration[1]' \in \UpwardClosure{\Configuration[1]}$,
  the roots of~$\Configuration[1]$ and~$\Configuration[1]'$
  can behave identically
  (as they see the same set of states),
  and if~$\Configuration[1]'$ has more leaves than~$\Configuration[1]$,
  then the additional leaves can copy
  the behavior of their indistinguishable siblings.
  So
  $\Configuration[1] \rightarrow^1 \Configuration[2]$
  implies that there is some
  $\Configuration[2]' \in \UpwardClosure{\Configuration[2]}$
  such that
  $\Configuration[1]' \rightarrow^1 \Configuration[2]'$.
  More precisely,
  let $\StarGraph, \StarGraph' \in \Stars$ be the underlying graphs
  of~$\Configuration[1]$ and~$\Configuration[1]'$,
  respectively.
  Since $\Configuration[1]' \succeq \Configuration[1]$,
  we know that the set of leaves of~$\StarGraph'$
  is a superset of the set of leaves of~$\StarGraph$.
  Let $\Selection$ be the selection of $\StarGraph$ underlying the step
  $\Configuration[1] \rightarrow^1 \Configuration[2]$.
  We now define the selection $\Selection'$ of $\StarGraph'$ as follows:
  \begin{itemize}
  \item The root~$r$ belongs to~$\Selection'$
    if and only if
    it belongs to~$\Selection$.
  \item For every state $q$:
    if $\Selection$ does not select any leaves in state $q$,
    then neither does $\Selection'$;
    otherwise,
    $\Selection'$ selects all leaves in state $q$ selected by $\Selection$,
    plus all other leaves in state $q$ that do not belong to~$\StarGraph$.
  \end{itemize}
  It follows that $\Selection' \supseteq \Selection$,
  and moreover a leaf of $\StarGraph'$ is selected in $\Selection'$
  only if some leaf of $\StarGraph$ in the same state is selected in $\Selection$.
  So a node of $\Selection'$ can only move to a state, say $q$,
  if some node of $\Selection$ also moves to $q$.
  Letting $\Configuration[2]'$ be
  the configuration reached by selecting $\Selection'$,
  this implies $\Configuration[2]' \in \UpwardClosure{\Configuration[2]}$,
  and thus
  $\UpwardClosure{\Configuration[1]} \rightarrow^1 \UpwardClosure{\Configuration[2]}$.

  For $i \geq 2$,
  the premise
  $\Configuration[1] \rightarrow^i \Configuration[2]$
  tells us that there exists a configuration~$\Configuration[3]$
  such that
  $\Configuration[1] \rightarrow^1
  \Configuration[3] \rightarrow^{i-1}
  \Configuration[2]$.
  By the induction hypothesis,
  this implies
  $\UpwardClosure{\Configuration[1]} \rightarrow^1
  \UpwardClosure{\Configuration[3]} \rightarrow^{i-1}
  \UpwardClosure{\Configuration[2]}$,
  and therefore
  $\UpwardClosure{\Configuration[1]} \rightarrow^i \UpwardClosure{\Configuration[2]}$.
  \hfill $\square$
  \smallskip

  As a direct consequence of Claim~1 we obtain:

  \smallskip
  \noindent\textbf{Claim~2.}\,
  If
  $\Set{\Configuration[1]} \rightarrow^* \UpwardClosure{\Configuration[2]}$,
  then
  ${\UpwardClosure{\Configuration[1]} \rightarrow^* \UpwardClosure{\Configuration[2]}}$. \\
  Indeed,
  $\Set{\Configuration[1]} \rightarrow^* \UpwardClosure{\Configuration[2]}$
  means that there is some
  $\Configuration[2]' \succeq \Configuration[2]$
  such that
  $\Configuration[1] \rightarrow^* \Configuration[2]'$.
  By Claim~1,
  it follows that
  $\UpwardClosure{\Configuration[1]} \rightarrow^* \UpwardClosure{\Configuration[2]'}$.
  Moreover,
  $\Configuration[2]' \succeq \Configuration[2]$
  implies
  $\UpwardClosure{\Configuration[2]'} \subseteq \UpwardClosure{\Configuration[2]}$.
  Therefore we get
  $\UpwardClosure{\Configuration[1]} \rightarrow^* \UpwardClosure{\Configuration[2]}$.
  \hfill $\square$
  \smallskip

  Claim~2 provides the motivation for the last notion we need to introduce:
  if we want to represent the set
  $\mathit{Pre}^*(\UpwardClosure{\Configuration[2]})$
  of \intro{predecessors} of $\UpwardClosure{\Configuration[2]}$
  (i.e., the configurations from which one can reach
  a configuration of $\UpwardClosure{\Configuration[2]}$ in zero or more steps),
  and if
  $\Configuration[1], \Configuration[1]' \in
  \mathit{Pre}^*(\UpwardClosure{\Configuration[2]})$
  such that
  $\Configuration[1] \prec  \Configuration[1]'$,
  then the representation of
  $\mathit{Pre}^*(\UpwardClosure{\Configuration[2]})$
  does not need to mention~$\Configuration[1]'$ explicitly,
  since
  $\Configuration[1] \in \mathit{Pre}^*(\UpwardClosure{\Configuration[2]})$
  already implies
  $\Configuration[1]' \in \mathit{Pre}^*(\UpwardClosure{\Configuration[2]})$.
  This leads us to represent
  $\mathit{Pre}^*(\UpwardClosure{\Configuration[2]})$
  by its set of minimal elements with respect to~$\preceq$.
  Formally,
  we define
  $\mathit{MinPre}^*(\UpwardClosure{\Configuration[2]})$
  to be the set of all configurations~$\Configuration[1]$
  such that
  $\Set{\Configuration[1]} \rightarrow^* \UpwardClosure{\Configuration[2]}$
  and there exists no configuration
  $\Configuration[1]' \prec \Configuration[1]$
  such that
  $\Set{\Configuration[1]'} \rightarrow^* \UpwardClosure{\Configuration[2]}$.

  \smallskip
  \noindent\textbf{Claim~3.}\,
  For every star configuration~$\Configuration[2]$,
  the set
  $\mathit{MinPre^\star}(\UpwardClosure{\Configuration[2]})$
  is finite. \\
  Since there are only finitely many base configurations,
  and every star configuration lies in
  the upward closure of its base configuration,
  it suffices to show that
  $\mathit{MinPre^\star}(\UpwardClosure{\Configuration[2]}) \cap
  \UpwardClosure{\Configuration[1]}$
  is finite
  for all $\Configuration[1] \in \mathit{Base}$.
  This follows easily from Dickson's Lemma,
  which states that for every infinite sequence
  $\vec{v}_1, \vec{v}_2, \ldots$
  of vectors of $\Naturals^k$,
  there exist two indices $i < j$
  such that
  $\vec{v}_i \leq \vec{v}_j$
  with respect to the pointwise partial order on vectors.
  Indeed,
  assume
  $\mathit{MinPre^\star}(\UpwardClosure{\Configuration[2]}) \cap
  \UpwardClosure{\Configuration[1]}$
  is infinite,
  and let
  $\Configuration_1, \Configuration_2, \ldots$
  be an enumeration of its elements,
  where
  $\Configuration_i = \Tuple{q, f_i}$.
  By Dickson's Lemma,
  there are $i < j$ such that
  $f_i(p) \leq f_j(p)$ for all $p \in Q$.
  This implies
  $\Configuration_i \preceq \Configuration_j$,
  and thus contradicts the minimality of $\Configuration_j$.
  \hfill $\square$
  \smallskip

  With all these notions in place,
  we can finally come back to the contradiction argument
  that proves Proposition~\ref{prop:CA*F-cannot-even-stars}.
  Let~$m$ be the maximum cardinality of any configuration
  that lies in the set
  $\mathit{MinPre^\star}(\UpwardClosure{\Configuration[2]})$
  of some base configuration~$\Configuration[2]$, i.e.,
  \begin{equation*}
    m \DefEq
    \max\bigSetBuilder{
      \Card{\Configuration[1]}
    }{
      \text{%
        there exists
        $\Configuration[2] \in \mathit{Base}$
        such that
        $\Configuration[1] \in \mathit{MinPre^\star}(\UpwardClosure{\Configuration[2]})$%
      }
    }.
  \end{equation*}
  Observe that $m$ is well-defined because
  $\mathit{Base}$ is finite by definition,
  and
  $\mathit{MinPre^\star}(\UpwardClosure{\Configuration[2]})$
  is finite by Claim~3.

  Now consider a star $\StarGraph_n$
  whose number of leaves~$n$ is chosen such that
  $n$ is even and
  $n \geq  (m \cdot |Q|)$,
  where~$|Q|$ is the number of states of~$A$.
  Let
  $\Run = \Tuple{\Configuration_0, \Configuration_1, \dots}$
  be a fair run of~$A$ on~$\StarGraph_n$.
  Since $n$ is even,
  $\Run$ is accepting,
  which means that there is a time $r \in \Naturals$
  such that for every $r' \geq r$,
  the configuration~$\Configuration_{r'}$ is accepting.
  Moreover,
  since the total number of configurations of~$A$ on~$G$ is finite,
  there is $s \geq r$
  such that the (accepting) configuration~$\Configuration_s$
  is visited infinitely often in~$\Run$.
  Since $A$ is strongly fair,
  no rejecting configuration is reachable from~$\Configuration_s$,
  because otherwise, by Lemma~\ref{lem:strongfairness},
  $\Run$ must visit that configuration.
  Let $\Configuration_s = \Tuple{q, f}$,
  and let~$p_{\text{max}}$ be a state
  that occurs maximally often at a leaf node of~$\Configuration_s$,
  i.e., $f(p_{\text{max}}) \geq f(p)$ for all $p \in Q$.

  Based on~$\Run$,
  we construct a fair run
  $\Run' = \Tuple{\Configuration_0', \Configuration_1', \dots}$
  of~$A$ on the star $\StarGraph_{n + 1}$
  such that
  the first $s + 1$ configurations
  $\Tuple{\Configuration_0', \dots, \Configuration_s'}$
  copy the behavior of~$\Run$.
  More precisely,
  the leaves $l_1, \ldots, l_n$ behave exactly as in~$\Run$.
  For the leaf~$l_{n + 1}$,
  let~$l_i$ be any of the leaves of~$\StarGraph_{n}$
  such that
  $\Configuration_s(l_i) = p_{\text{max}}$.
  During the first $s$ steps,
  the schedule of $\Run'$ selects~$l_{n+1}$
  if and only if
  the schedule of~$\Run$ selects~$l_i$.
  It follows that $l_{n+1}$ visits the same sequence of states as~$l_i$,
  and so
  $\Configuration_s'(l_{n+1}) = p_{\text{max}}$.
  Note that
  this construction does not contradict the strong fairness constraint
  because we only fix a finite prefix of~$\Run'$.
  We now extend~$\Run'$ in such a way
  that it satisfies the strong fairness constraint.

  Since $n + 1$ is odd,
  the run~$\Run'$ must eventually visit only rejecting configurations.
  In particular,
  some rejecting configuration~$\Configuration_t'$
  is reachable from~$\Configuration_s'$,
  and so
  $\Configuration_s'  \succeq \Configuration[2]$
  for some
  $\Configuration[2] \in
  \textit{MinPre}^\star(\UpwardClosure{\operatorname{base}(\Configuration_t')})$.

  \smallskip
  \noindent\textbf{Claim~4.}\,
  $\Configuration_s  \succeq \Configuration[2]$. \\
  Recall that
  $\Configuration_s = \Tuple{q, f}$,
  and let
  $\Configuration_s' = \Tuple{q, f'}$
  and
  $\Configuration[2] = \Tuple{q, g}$.
  We have to show that $f(p) \geq g(p)$ for every state $p \in Q$.
  To do so,
  we distinguish two cases:
  \begin{itemize}
  \item If $p \neq p_{\text{max}}$,
    then by the definition of $\Configuration_s'$,
    we have $f(p) = f'(p)$,
    and since
    $\Configuration_s' \succeq \Configuration[2]$,
    it follows immediately that
    $f(p) \geq g(p)$.
  \item If $p = p_{\text{max}}$,
    then by the pigeonhole principle
    and the definitions of~$n$ and~$p_{\text{max}}$,
    we have $f(p) \geq n/|Q| \geq m$.
    Moreover,
    we have $g(p) \leq m$
    because the definition of~$m$
    ensures that $\Card{\Configuration[2]} \leq m$.
    Hence,
    $f(p) \geq g(p)$.
    \hfill $\square$
  \end{itemize}

  Since
  $\Configuration[2] \in
  \textit{MinPre}^\star(\UpwardClosure{\operatorname{base}(\Configuration_t')})$,
  Claim~4 tells us
  that $\Configuration_s$ can also reach some rejecting configuration
  in $\UpwardClosure{\operatorname{base}(\Configuration_t')}$.
  This contradicts what we have established above.
  We therefore conclude that
  \Type{set}{stabilizing}{liberal}{strong}-automata
  cannot recognize $\EvStars$,
  and by Theorem~\ref{thm:simulate-exclusivity-with-fairness},
  the same holds for \Type{set}{stabilizing}{exclusive}{strong}-automata.
\end{proof}

\subsection{Proofs of Section \ref{sec:pprotocols}}

\PropPopulationProtocols*

\begin{proof}
  We present a simulation
  that runs a graph population protocol on a distributed automaton.
  To this end,
  the automaton has to simulate a scheduler that selects
  ordered pairs of adjacent nodes instead of arbitrary sets of nodes.
  For any pair $\Tuple{u, v}$ that is selected to perform a transition,
  let us call~$u$ the \emph{initiator} and~$v$ the \emph{responder}
  of the transition.
  By Theorem~\ref{thm:simulate-exclusivity-with-fairness},
  we may assume that
  the automaton's scheduler selects a single node in each step.

  The main idea of the construction is as follows:
  When a node $u$ is selected and sees that
  it can become the initiator of a transition,
  it declares its intention to do so by raising the flag~“?”.
  Then $u$ waits until
  some neighbor~$v$ is selected and raises the flag~“!”,
  which signals that $v$ wants to become the responder of a transition.
  If this happens,
  the next time $u$ is selected,
  it computes its new state according to
  the state of $v$ and the transition function of the population protocol,
  but also keeps its old state in memory
  so that $v$ can still see it.
  After that,
  $v$~also updates its state,
  and finally $u$ deletes its old state,
  which completes the transition.
  Throughout this protocol,
  the nodes verify that they have exactly one partner during each transition.
  If this condition is violated,
  they raise the error flag~“$\bot$” and abort their current transition.

  Formally,
  let $\varPi = \Tuple{Q, \delta_0, \delta, Y, N}$
  be a population protocol on $\Alphabet$-labeled graphs.
  We construct the \Type{multiset}{stabilizing}{exclusive}{strong}-automaton~$A$
  with machine
  $M = \Tuple{Q', \delta'_0, \delta', Y', N'}$,
  where
  \begin{equation*}
    Q' =\, Q \,\cup\, (Q \times \Set{?, !, \bot}) \,\cup\, Q^2,
  \end{equation*}
  the sets~$Y'$ and $N'$ are defined analogously,
  $\delta'_0(a) = \delta_0(a)$ for all $a \in \Alphabet$,
  and $\delta'$~is defined as follows.
  Let~$v$ be the node currently selected by the scheduler.

  \begin{enumerate}
  \item In case~$v$ is in state $q \in Q$:
    \begin{enumerate}
    \item \label{itm:request}
      if all of $v$'s neighbors are in states of~$Q$,
      then $v$ moves to $\Tuple{q, ?}$;
    \item \label{itm:accept}
      if exactly one of $v$'s neighbors is
      in some state of $Q \times \Set{?}$
      and all others are in states of~$Q$,
      then $v$ moves to $\Tuple{q, !}$;
    \item \label{itm:deadlock}
      if several of $v$'s neighbors are in states of $Q \times \Set{?}$,
      then $v$ moves to $\Tuple{q, \bot}$;
    \item \label{itm:wait-silence}
      otherwise, $v$ remains in state~$q$.
    \end{enumerate}
    Intuitively,
    in rule~\ref{itm:request},
    $v$ makes a request for a transition partner,
    in rule~\ref{itm:accept},
    $v$ accepts the request of some other node,
    and in rule~\ref{itm:deadlock},
    $v$ signals an error
    because it has received multiple requests.
    Signaling the error is necessary to guarantee that
    two requesting nodes with a common neighbor
    do not end up in a deadlock.
    In rule~\ref{itm:wait-silence},
    $v$ simply waits for ongoing transitions in its neighborhood to be completed.
  \item In case~$v$ is in state $\Tuple{q, ?}$:
    \begin{enumerate}
    \item \label{itm:wait-accept}
      if all of $v$'s neighbors are in states of~$Q$,
      then $v$ remains in $\Tuple{q, ?}$;
    \item \label{itm:transition-init}
      if exactly one of $v$'s neighbors is in
      a state of the form $\Tuple{p, !}$
      and all others are in states of~$Q$,
      then $v$ moves to $\Tuple{q, \delta(q, p)}$;
    \item \label{itm:abort-request}
      otherwise, $v$ moves to $\Tuple{q, \bot}$.
    \end{enumerate}
    Intuitively,
    in rule~\ref{itm:wait-accept},
    $v$ waits for some node to accept its request,
    in rule~\ref{itm:transition-init},
    $v$ initiates a transition of~$\varPi$
    with the unique responder that has accepted its request,
    and in rule~\ref{itm:abort-request},
    $v$~aborts its attempt to make a transition.
    The latter happens either
    if some neighbor of~$v$ has received multiple requests,
    or if several nodes have accepted $v$'s request
    (in which case $v$'s new state informs those nodes of the error).
  \item In case~$v$ is in state $\Tuple{q, !}$:
    \begin{enumerate}
    \item \label{itm:wait-confirmation}
      if exactly one of $v$'s neighbors is
      in some state of $Q \times \Set{?}$
      and all others are in states of~$Q$,
      then $v$ remains in $\Tuple{q, !}$;
    \item \label{itm:transition-proceed}
      if exactly one of $v$'s neighbors is in
      a state of the form $\Tuple{p, p'}$
      and all others are in states of~$Q$,
      then $v$ moves to $\delta(p, q)$;
    \item \label{itm:abort-accept}
      otherwise, $v$ moves to state~$q$.
    \end{enumerate}
    Intuitively,
    in rule~\ref{itm:wait-confirmation},
    $v$ waits for its potential transition partner
    to initiate the transition,
    in rule~\ref{itm:transition-proceed},
    $v$ performs its own part of the transition,
    and in rule~\ref{itm:abort-accept},
    $v$~aborts the transition attempt.
    The latter happens if the initiator of the transition signals an error.
  \item In case~$v$ is in state $\Tuple{q, \bot}$:
    \begin{enumerate}
    \item \label{itm:wait-error}
      if some neighbor of~$v$ is in a state of $Q \times \Set{?, !}$,
      then $v$ remains in $\Tuple{q, \bot}$;
    \item \label{itm:resume}
      otherwise, $v$ moves to state~$q$.
    \end{enumerate}
    Intuitively,
    in rule~\ref{itm:wait-error},
    $v$ waits for its affected neighbors to see that an error has occurred,
    and in rule~\ref{itm:resume},
    $v$ returns to the state it had before the last failed transition attempt.
  \item In case~$v$ is in state $\Tuple{q, q'} \in Q^2$:
    \begin{enumerate}
    \item \label{itm:wait-transition}
      if some neighbor of~$v$ is in a state of $Q \times \Set{!}$,
      then $v$ remains in $\Tuple{q, q'}$;
    \item \label{itm:transition-complete}
      otherwise, $v$ moves to state~$q'$.
    \end{enumerate}
    Intuitively,
    in rule~\ref{itm:wait-transition},
    $v$ waits for its transition responder to perform its part of the transition;
    to make this possible,
    $v$ must still keep its old state~$q$ in memory.
    In rule~\ref{itm:transition-complete},
    the transition has been completed,
    so $v$ can remove its old state.
  \end{enumerate}

  By Lemma~\ref{lem:strongfairness},
  strong fairness guarantees that every ordered pair of nodes
  will be able to perform a transition infinitely often,
  and more generally,
  every finite sequence of pairs will be selected infinitely often
  by the simulated scheduler.
  Moreover,
  if several pairs make transitions simultaneously,
  the construction ensures that none of these pairs have a node in common.
  This means that the outcome of the transitions would not change
  if they were rescheduled sequentially.
  Hence,
  every fair run of automaton~$A$ simulates
  a fair run of population protocol~$\varPi$,
  and since~$\varPi$ satisfies the consistency condition,
  so does~$A$.
  Therefore the two devices are equivalent.

  Notice that the above construction relies on the fact that
  $A$ is a \Type{multiset}{stabilizing}{exclusive}{strong}-automaton:
  nodes must be able to count to verify that
  they have exactly one partner during each transition;
  acceptance by stable consensus and strong fairness
  are required to match the way population protocols are executed;
  and just as in the proof of Proposition~\ref{prop:CasF-recog-stars},
  exclusive selection is used to simplify the design of the automaton.
  In particular,
  when a responder accepts the request of an initiator
  (rule~\ref{itm:accept}),
  it is guaranteed that
  none of its other neighbors make a new request at the same time.
  Similarly,
  when a node initiates a transition with a responder
  (rule~\ref{itm:transition-init}),
  it can be sure that
  its request is not simultaneously accepted by another node.
\end{proof}

\subsection{Proofs of Section \ref{sec:conclusions}}

\PropAFkSimulatesCAF*

\newcommand{\fc}{\textit{fc}}
\newcommand{\sco}{\textit{sc}}
\newcommand{\NE}{\textit{NE}}
\newcommand{\ConfigurationWC}{\Configuration_{wc}}

\begin{proof}
Given a \Type{multiset}{stabilizing}{*}{strong}-automaton $A$, we have to describe a \Type{set}{stabilizing}{*}{strong}-automaton \(B\) such that for every graph~$G$ of maximum degree~$k$, every fair run of $B$ on $G$ simulates some fair run of $A$ on $G$. Observe that this is enough to prove that $A$ and $B$ are equivalent on graphs of maximum degree~$k$. Indeed, since by assumption $A$ satisfies the consistency condition, either all fair runs of $A$ on $G$ are accepting, or all are rejecting. If every fair run of $B$ on $G$ simulates some fair run of $A$ on $G$, then $B$ also satisfies the consistency condition and accepts $G$ if{}f $A$ accepts $G$. 

In the following,
we construct a \Type{set}{stabilizing}{liberal}{strong}-automaton~$B$ that simulates
a \Type{multiset}{stabilizing}{liberal}{strong}-automaton~$A$ on any graph of maximum degree~$k$.
(The same construction can also be used to go from \Type{multiset}{stabilizing}{exclusive}{strong}-automata to \Type{set}{stabilizing}{exclusive}{strong}-automata.)

Let $Q$ be the set of states of $A$. A state of $B$ is a fivetuple $\alpha =(q_0, q, p, \fc, \sco)$, where $q_0, q \in Q$ are the \emph{initial} and \emph{current state}, respectively, $p \in \{0,1,2\}$ is the \emph{phase}, and $\fc \in \Range{k^2}$ is the \emph{first color}, and $\sco \in \{0,1\}$ is the \emph{second color}, respectively.

Let $G = \Tuple{V, E, \lambda}$ be a graph of maximum degree~$k$. The initial state of a node $v$ of $G$ in $B$ is $(q_0,q_0,0,0,0)$, where $q_0 = \delta_0(\lambda(v))$ and $\delta_0$ is the initialization function of $A$. 
Let us now give a more precise but still intuitive description of the intended meaning of ``a node $v$ of a graph $G$ is currently in state $\alpha=(q_0, q, p, \fc,\sco)$''.  The first two components are straightforward:

\begin{itemize}
\item $q_0$ is always $\delta_0(\lambda(v))$. (That is, the transition function of $B$, introduced below, never changes the first component of a state.) Sometimes the node needs to go back to its initial state, and this component just tells the node where to go.
\item $q$ is the current state of $v$ in the run of $A$ being simulated.
\end{itemize}

The other three components require some further explanation. Given a node $v$, let $\NE(v)$ be the set containing $v$ and its neighbors. We say that a configuration $\Configuration$ is \emph{well colored} if for every node $v$ the first colors of $v$ and all its neighbors are pairwise distinct in $\Configuration$ (i.e., each first color occurs at most once in $v$'s neighborhood). A goal of the protocol is to eventually reach a well-colored configuration $\ConfigurationWC$ such that from then on no node ever changes its first color.  Intuitively, the first color of a node at $\ConfigurationWC$ becomes its \emph{locally unique identity}: an identifier that never changes, different from the identities of all its neighbors and neighbors' neighbors. With locally unique identities the nodes can then easily simulate the moves of $A$: Indeed, in order to know how many neighbors they have in a state of $A$, say $q_1$, they just count the number of different states they see of the form $(q_0, q_1, p, \fc,\sco)$.

To achieve this goal, the protocol uses the second colors. In phase 0 the nodes restart their states (initially this is superfluous because they are already there), and move to phase~1. In phase 1, the nodes select an arbitrary distribution of first colors. Since the nodes are deterministic, they rely on strong fairness to ensure that eventually a well-colored distribution is chosen. The nodes then move to phase 2, where they start simulating $A$ under the assumption that the current configuration is well colored.
However, at the same time they keep changing their second colors, and start to watch out for neighbors with the same first color as themselves, and for pairs of neighbors with the same first color but distinct second colors. Whenever they detect one of these two situations, they know that their assumption was incorrect, which implies that the simulation they have carried out so far is useless. So they move back to phase 0. We recall that, as in some other proofs, the nodes do not move synchronously from phase to phase; instead, a node moves to a new phase, and waits for its neighbors to follow.

Let us now describe the transition function of $B$. Let $\Configuration$ denote the current configuration of $B$. Fix a node $v$ of $G$, and  let $\alpha=(q_0, q, p, \fc,\sco)$ be the current state of $v$ in $\Configuration$. Further, let $q'$ be the state $v$ would move to in machine $A$ from the configuration of $A$ corresponding to $\Configuration$. Finally, let $(\fc +1)$ denote $(\fc+1) \bmod (k^2+1)$, and  
$(p+1)$ and $(p-1)$ denote $(p+1) \bmod 3$ and $(p-1) \bmod 3$, respectively. If $v$ is selected by the scheduler at $\Configuration$ , then its next state is determined as follows:

\begin{description}
\item[(0)] If $v$ is in phase 0 then:
\begin{description}
\item[(0.a)] If some neighbor of $v$ is in phase 2, then $v$ stays in $\alpha$.
\item[(0.b)] If all neighbors of $v$ are in phase 0 or 1, then $v$ moves to 
$\alpha[q \rightarrow q_0, p \rightarrow 1]$.
\end{description}
\item[(1)] If $v$ is in phase 1 then:
\begin{description}
\item[(1.a)] If at least one neighbor of $v$ is in phase 0, then $v$ moves to $\alpha[\fc \rightarrow \fc+1]$; \\
(Intuitively, $v$ waits for its neighbors in phase 0 to catch up.)
\item[(1.b)] If all neighbors of $v$ are in phase 1, then $v$ moves to 
$\alpha[p \rightarrow 2, \fc \rightarrow \fc+1]$; \\
(The node initiates a new phase.)
\item[(1.c)] If at least one neighbor of $v$ is in phase 2, then $v$ moves to $\alpha[p \rightarrow 2]$. 
\end{description}
\item[(2)] If $v$ is in phase 2 then:
\begin{description}
\item[(2.a)] If some neighbor of $v$ is in phase 1, then $v$ moves to $\alpha[\fc \rightarrow \fc+1]$. 
\item[(2.b)] If all neighbors of $v$ are in phase 2, and any two nodes of $\NE(v)$ with the same first color also have the same second color, then $v$ moves to $\alpha[q\rightarrow q',\sco \rightarrow 1- \sco]$. \\
(In this case $v$ sees no local violation of the well-coloring condition, and so it simulates a move of $A$, and changes its second color.)
\item[(2.c)] If all neighbors of $v$ are in phase 2, and $\NE(v)$ contains two nodes with the same first color but distinct second colors, then $v$ moves to $\alpha[p \rightarrow 0]$; 
\item[(2.d)] If some neighbor of $v$ is in phase 0, then $v$ moves to $\alpha[p \rightarrow 0]$.
\end{description}
\end{description}

This concludes the description of $B$. In the rest of the proof we show that  $B$ is a distributed automaton, i.e., that it satisfies the consistency condition, and that every fair run of $B$ on $G$ simulates some fair run of $A$ on $G$. The proof is in four steps.

\smallskip\noindent\textbf{Claim 1.} Every run of $B$ eventually reaches a well-colored configuration with all nodes in phase 2. \\
By strong fairness and Lemma~\ref{lem:strongfairness}, it suffices to show that for every configuration there exists a finite sequence of selections such that the configuration reached after executing them is well colored with all nodes in phase 2. First we show that it is possible to color the nodes of $G$ with at most $k^2+1$ different colors so that the colors of every set of nodes $\NE(v)$ are pairwise distinct. Let $G'$ be the result of triangulating $G$, i.e., adding an edge $\{v_1, v_3\}$ for every pair of edges $\{v_1, v_2\}, \{v_2, v_3\} \in G$ such that $v_1 \neq v_3$. Since $G$ has maximum degree~$k$, the graph $G'$ has maximum degree at most $k^2$. Clearly, a coloring of $G'$ in the usual graph-theoretical sense (i.e., for every edge $\{v_1, v_2\}$ of $G'$ the nodes $v_1$ and $v_2$ have different colors) satisfies that the colors of every set $\NE(v)$ in $G$ are pairwise distinct. So it suffices to exhibit a coloring of $G'$ with $k^2+1$ colors. Such a coloring can be obtained by applying the standard greedy algorithm that produces a coloring of a graph with maximum degree $m$ using $m+1$ colors (in our case $m = k^2$). 

We prove the existence of a reachable well-colored configuration with all nodes in phase 2 in two steps:
\begin{description}
\item[(1)] Every reachable configuration can reach either a well-colored configuration with all nodes in phase 2, or a configuration with all nodes in phase 0. \\
Let $\Configuration$ be a reachable configuration. Inspection of (0)-(2) shows that from 
$\Configuration$ we can reach~$\Configuration'$ with all nodes in phase 2. If $\Configuration'$
is well colored we are done. Otherwise, there is a node $v$ such that two nodes of $\NE(v)$ have the same first color in $\Configuration'$. If these nodes have distinct second colors, we can select $v$ and bring it to phase 0 with (2.c), and then (2.d) yields the result. If the nodes have the same second colors, we select one of them. If  (2.b) applies, then its second color changes, and
we can select $v$ as before. If (2.c) applies, then this node moves to phase 0, and then (2.d) yields the result. 
\item[(2)] Every configuration with all nodes in phase 0 can reach a well-colored configuration with all nodes in phase 2. \\
Take a spanning tree $T$ of $G$. Starting with $T' := T$, repeatedly select a leaf $v$ of $T'$ as many times as necessary to give it any first color we wish (this is possible by (0.b) and (1.a)); we then remove $v$ from $T'$ and iterate. When $T'$ consists of just one node, we proceed similarly, but using (1.b) and (2.a). This yields a well-colored configuration with one node in phase 2
and all others in phase 1. We repeatedly select nodes in phase 1 with a neighbor in phase 2 and apply (1.c).
\hfill $\square$
\end{description}

\smallskip\noindent\textbf{Claim 2.} The set of well-colored configurations with all nodes in phase 2 is closed under the transition relation. \\
In such configurations only (2.b) is enabled, which changes neither the phase nor the first color of a node.
So after any transition the new configuration is also well-colored, and all nodes stay in phase 2.
\hfill $\square$


\newcommand{\DiffCol}{\mathit{DC}}
\newcommand{\RunA}{\Run^A}
\newcommand{\RunB}{\Run^B}
\newcommand{\ScheduleA}{\Schedule^A}
\newcommand{\ScheduleB}{\Schedule^B}
\newcommand{\ScheduleC}{\Schedule^{\mathit{aux}}}
\newcommand{\SelectionA}{\Selection^A}
\newcommand{\SelectionB}{\Selection^B}
\newcommand{\SelectionC}{\Selection^{\mathit{aux}}}
\newcommand{\NoDuplicatesRunA}{\widehat{\Run}^A}
\newcommand{\NoDuplicatesRunB}{\widehat{\Run}^B}
\newcommand{\ConfigurationA}{\Configuration^A}
\newcommand{\ConfigurationB}{\Configuration^B}
\newcommand{\state}[3]{q^{#1}_{#2#3}}

\smallskip

Let us now prove that $B$ satisfies the consistency condition, and that
it is equivalent to~$A$ on graphs of maximum degree~$k$.
Let  $\RunB = \Tuple{\ConfigurationB_0, \ConfigurationB_1, \ConfigurationB_2, \dots}$ be an arbitrary strongly fair run of $B$ on $G$. It suffices to show that there exists a strongly fair run $\RunA$ of $A$ on $G$ such that $\RunB$ is accepting if{}f $\RunA$ is accepting. Indeed, since $A$ satisfies the consistency condition by hypothesis, it follows that $B$ is also consistent, and that $B$ accepts $G$ if{}f $A$ does, which implies the equivalence of~$A$ and~$B$ on $k$-bounded graphs.

Let $\ScheduleB = (\SelectionB_0, \SelectionB_1, \SelectionB_2, \ldots) \in (2^V)^\omega$ be a schedule that schedules $\RunB$. We now define a schedule $\ScheduleA=(\SelectionA_0, \SelectionA_1, \SelectionA_2, \ldots)$, and then choose $\RunA$ as the run scheduled by $\ScheduleA$. For every node $v$, let $t_v$ be the smallest time after which $v$ and its neighbors reach phase 2 and stay in it forever (in run~$\RunB$), which exists by Claims 1 and 2.  For every $t \in \Naturals$, we decide whether $v \in \SelectionA_t$ or not as follows: 
\begin{quote}
If $t \leq t_v$, then $v \notin \SelectionA_t$; if $t > t_v$, then $v \in \SelectionA_t$ if{}f $v \in \SelectionB_t$. 
\end{quote}
\noindent So, intuitively, in $\ScheduleA$ a node $v$ is never selected before $\NE(v)$ has ``stabilized'', and after that it is selected whenever $\ScheduleB$ selects it. It remains to show that $\RunA$ is strongly fair, and that $\RunA$ is accepting if{}f $\RunB$ is accepting.

\smallskip\noindent\textbf{Claim 3.} $\RunA$ is strongly fair. \\
By Claims 1 and 2 and the definition of $\ScheduleA$, there is a time $t$ such that $\SelectionA_{t'}= \SelectionB_{t'}$ for every $t' \geq t$ (intuitively, $t$ is the time at which all nodes have stabilized in phase 2). Since $\ScheduleB$ is strongly fair by hypothesis, and strong fairness is independent of the properties of any finite prefix, $\ScheduleA$ is also strongly fair. So $\RunA$ is strongly fair.
\hfill $\square$

\smallskip\noindent\textbf{Claim 4.} $\RunA$ is accepting if{}f $\RunB$ is accepting. \\
Let $\RunA = \Tuple{\ConfigurationA_0, \ConfigurationA_1, \ConfigurationA_2, \dots}$,
and let $v$ be an arbitrary node of~$G$. 
It suffices to prove that $\ConfigurationA_t(v) = \ConfigurationB_t(v)$ holds for every $t \geq t_v$. (Indeed, by definition a run is accepting if{}f  every node eventually visits accepting states only, and so, since  $\ConfigurationA_t(v) = \ConfigurationB_t(v)$ for every $t \geq t_v$, this holds for $\RunA$ if{}f it holds for $\RunB$.) We proceed by induction on $t$. 

\noindent \textbf{Base:} $t=t_v$. Let $q_{0v}$ be the initial state of $v$. We prove $\ConfigurationA_{t_v}(v) = q_{0v} = \ConfigurationB_{t_v}(v)$. We have $\ConfigurationA_t(v) = q_{0v}$ for every $t \leq t_v$ because $v \notin \SelectionA_t$ for any $t \leq t_v$.
Moreover, we have $\ConfigurationB_{t_v}(v) = q_{0v}$ because $v$ moves to $q_{0v}$ the last time it moves to phase 1 (case (0.b)), and stays in  $q_{0v}$ until it and all its neighbors reach phase 2 (case (2.b)). But this is precisely the time $t_v$: Since $v$ never leaves phase 2 again, neither do its neighbors (otherwise they would ``drag'' $v$ to phase~$0$ with them).

\noindent \textbf{Step:} $t > t_v$. By induction hypothesis we have $\ConfigurationA_{t-1}(v) = \ConfigurationB_{t-1}(v)$, and by the definition of~$\ScheduleA$ we have $v \in \SelectionA_t$ if{}f $v \in \SelectionB_t$. So it suffices to show 
$\ConfigurationA_{t-1}(u) = \ConfigurationB_{t-1}(u)$ for every neighbor~$u$ of~$v$. 
Fix a neighbor $u$. Consider two cases:
\begin{itemize}
\item $t \geq t_u$. Then $\ConfigurationA_{t-1}(u) = \ConfigurationB_{t-1}(u)$ follows from the induction hypothesis applied to the node $u$.
\item $t < t_u$. Let $q_{0u}$ be the initial state of $u$. Since, by definition, $\ScheduleA$ never selects $u$ before time $t_u$, we have $\ConfigurationA_{t-1}(u) = q_{0u}$. We show $\ConfigurationB_{t-1}(u) = q_{0u}$. Since $t < t_u$ holds but $u$ will never leave phase 2 after $t$ by hypothesis, some neighbor of $u$ will still change its phase after $t$. So its neighbor is in phase 1. But all nodes in phase 1 are in their initial state.  
\qedhere
\end{itemize}
\end{proof}


\end{document}